\documentclass[aps,pra,superscriptaddress,notitlepage, amsfonts,longbibliography, twocolumn]{revtex4-2}

\usepackage{qcircuit}
\usepackage[dvipdfmx]{graphicx}
\usepackage{amsmath,amssymb,amsthm,mathrsfs,amsfonts,dsfont}
\usepackage{subfigure, epsfig}%, subcaption}
\usepackage{braket}
\usepackage{bm}
\usepackage{enumerate}
\usepackage{physics}
\newtheorem{thm}{Theorem}

\newtheorem{cor}{Corollary}
\usepackage[colorlinks,linkcolor=blue,citecolor=blue]{hyperref}
\usepackage{stmaryrd}
\renewcommand{\thetable}{\arabic{table}}

\newcommand{\red}[1]{\textcolor{black}{#1}}

\makeatletter
\renewcommand*{\numberline}[1]{\hb@xt@1em{#1\hfil}} 
\makeatother

\graphicspath{{./fig/}}

\begin{document}

\title{Symmetric Clifford twirling for cost-optimal quantum error mitigation \\in early FTQC regime}

% ------------  AUTHORS AND AFFILIATIONS ----------
\author{Kento Tsubouchi}
\email{tsubouchi@noneq.t.u-tokyo.ac.jp}
\affiliation{Department of Applied Physics, The University of \mbox{Tokyo, 7-3-1} Hongo, Bunkyo-ku, Tokyo 113-8656, Japan}

\author{Yosuke Mitsuhashi}
\affiliation{Department of Basic Science, The University of \mbox{Tokyo, 3-8-1} Komaba, Meguro-ku, Tokyo 153-8902, Japan}

\author{Kunal Sharma}
\affiliation{IBM Quantum, IBM T.J. Watson Research Center, Yorktown Heights, NY 10598, USA}

\author{Nobuyuki Yoshioka}
\email{nyoshioka@ap.t.u-tokyo.ac.jp}
\affiliation{Department of Applied Physics, The University of \mbox{Tokyo, 7-3-1} Hongo, Bunkyo-ku, Tokyo 113-8656, Japan}
\affiliation{\mbox{International Center for Elementary Particle Physics, The University of Tokyo, 7-3-1 Hongo, Bunkyo-ku, Tokyo 113-0033, Japan}}
\affiliation{JST, PRESTO, 4-1-8 Honcho, Kawaguchi, Saitama, 332-0012, Japan}

\begin{abstract}
Twirling noise affecting quantum gates is essential in understanding and controlling errors, but applicable operations to noise are usually restricted by symmetries inherent in quantum gates.
In this work, we propose symmetric Clifford twirling, a Clifford twirling utilizing only symmetric Clifford operators that commute with certain Pauli subgroups.
We fully characterize how each Pauli noise is converted through the twirling and show that certain Pauli noise can be scrambled to a noise exponentially close to the global white noise.
Moreover, we provide numerical demonstrations for highly structured circuits, such as Trotterized Hamiltonian simulation circuits, that noise effect on typical observables can be described by the global white noise.
We further demonstrate that symmetric Clifford twirling and its hardware-efficient variant using only local symmetric Clifford operators can significantly accelerate the scrambling.
These findings enable us to mitigate errors in non-Clifford operations with minimal sampling overhead in the early \red{fault-tolerant regime}.
\end{abstract}

\maketitle

\section*{Introduction}
As a powerful countermeasure for errors affecting quantum computers, fault-tolerant quantum computing (FTQC) using quantum error correction has been studied in recent decades~\cite{shor1995scheme, knill1996threshold, aharonov1997fault, lidar2013quantum}.
Despite significant experimental advances achieving the break-even point for error correction on multiple platforms~\cite{ofek2016extending,  krinner2022realizing, sivak2023real, google2023suppressing, bluvstein2024logical, acharya2024quantum}, the early generations of FTQC are still expected to be subject to a considerable amount of residual noise, due to the high overhead required for fully fledged FTQC. Consequently, developing techniques to eliminate the remaining errors in logical qubits is crucial.

One of the leading candidates is to employ the quantum error mitigation (QEM) techniques.
The goal of QEM is to predict the expectation value of an error-free quantum circuit by combining the output from error-prone quantum circuits, in exchange for increased circuit executions~\cite{temme2017error, li2017efficient, endo2021hybrid, cai2023quantum, kim2023evidence}.
QEM methods called probabilistic error cancellation~\cite{temme2017error, endo2018practical, van2023probabilistic} are known to effectively mitigate logical errors and thereby reduce the required code distance for logical qubits~\cite{piveteau2021error, lostaglio2021error, suzuki2022quantum}.
However, \red{it has recently been shown that} this method is suboptimal: the scaling of the sampling overhead is quadratically worse compared to the theoretical lower bound~\cite{tsubouchi2023universal}.
This implies that a cost-optimal QEM method, if available, could permit a logical error rate up to twice as large, allowing for a smaller code distance in logical qubits. 
Such an improvement could substantially reduce hardware requirements, especially in regimes where physical error rates are comparable to the threshold value for error-correcting codes.
Therefore, it is imperative to employ a cost-optimal QEM method that saturates the lower bound.

In fact, cost-optimal QEM becomes feasible under a specific condition: when the logical error can be characterized as global white noise (also known as global depolarizing noise), we can cost-optimally mitigate errors by simply rescaling the noisy expectation value~\cite{tsubouchi2023universal}.
Although white noise has been argued to arise under random circuit sampling~\cite{arute2019quantuma, dalzell2021random, deshpande2022tight, morvan2023phase}, it remains an open question how we can reliably ensure or accelerate the convergence to white noise.

A straightforward approach to converting noise to global white noise is Clifford twirling: by randomly inserting global Clifford operations before and after the noise channel, we can scramble the noise to global white noise~\cite{emerson2005scalable, dankert2009exact}.
This approach, however, is not practical for most non-Clifford gates since the noise channels cannot be divided from the target operations. In other words, it is impossible to insert additional Clifford operations in between, thus prohibiting full Clifford twirling.
One solution is to utilize only the Clifford operations that commute with the non-Clifford gates, as proposed for Pauli twirling in Ref.~\cite{kim2023scalable}.
However, the method only considers local operations for a single Pauli rotation and does not accelerate white noise approximation.
In order to fully exemplify the early FTQC scheme, it is an urgent task to establish a unified understanding and methodology regarding the full symmetric Clifford operations.

\begin{figure*}[t]
    \begin{center}
        \includegraphics[width=1.0\linewidth]{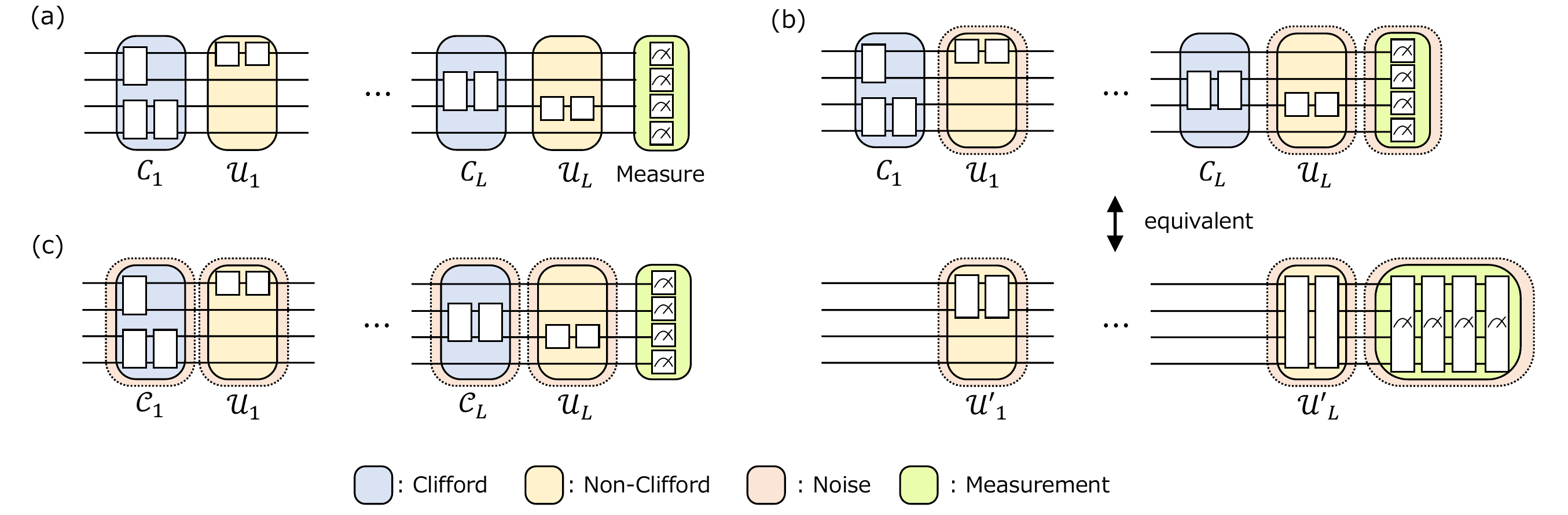}
        \caption{Graphical representation of logical quantum circuit structures in the early FTQC regime. Each panel indicates (a) a noiseless ideal circuit, (b) a noisy circuit under a single-thread supply of magic states (see \red{Supplementary Note~1} for details), and (c) a noisy circuit under multi-thread supply of magic states. 
        In (b), the noise in the Clifford layer is not the target of twirling, since we actually implement a conjugated circuit in which Clifford operations are absorbed into the Pauli measurement.
        }
        \label{fig_earlyFTQC_circuit}
    \end{center}
\end{figure*}

In this work, we bridge this gap by proposing \textit{symmetric Clifford twirling}, a Clifford twirling using symmetric Clifford operators~\cite{mitsuhashi2023clifford} that commute with certain Pauli subgroups.
By appropriately choosing the Pauli subgroup, we obtain symmetric Clifford operators that commute also with non-Clifford operations, allowing us to twirl its noise.
We completely characterize how Pauli noise channels are converted through symmetric Clifford twirling and show that some Pauli noise can be scrambled to the noise exponentially close to global white noise, enabling the cost-optimal QEM.
In addition, we propose \textit{$k$-sparse symmetric Clifford twirling}, a symmetric Clifford twirling using Clifford operators acting on up to $k$ qubits.
We show that this simplified twirling can still scramble the noise polynomially close to global white noise.

Furthermore, we numerically show for highly structured circuits that noise effects on the expectation values of typical observables are well described by the global white noise.
By applying our techniques to the Trotterized Hamiltonian simulation circuits, we show that such a scrambling effect can be accelerated by ($k$-sparse) symmetric Clifford twirling.
We also numerically verify the robustness of our protocol against noise on twirling gadgets.
These findings enable us to mitigate errors in non-Clifford operations, especially the Pauli rotation gates, with minimal sampling overhead in the early FTQC regime.

\red{\section*{Results}}
%\subsection*{Problem Setup}
Let us begin by introducing a logical quantum circuit structure that encapsulates the essence of the early FTQC regime.
The noiseless $n$-qubit logical circuit shown in Fig.~\ref{fig_earlyFTQC_circuit}(a) consists of $L$ alternating sequences of a Clifford operation $\mathcal{C}_l$ and a non-Clifford layer $\mathcal{U}_l$, where the subscript $l$ indexes each layer.
\red{In the early stages of FTQC, it is expected that the size of magic state factories is limited, or magic states are supplied with small success probabilities so that the effective number of supply is small even if there are multiple factories.
In these cases, the non-Clifford operations implementable per single time step are limited~\cite{litinski2019game}.
Therefore, even though we have an identical noiseless logical circuit we wish to implement,} the appropriate noisy logical circuit structure varies depending on the available number of magic state factories.
Specifically, Figs.~\ref{fig_earlyFTQC_circuit}(b) and (c) illustrate the cases of single-thread and multi-thread supplies of magic states, respectively.

\red{When the (effective) number of magic state factories is small, we have access to only a single magic state per time step, which we call a single-thread supply of magic.
In this case,} it is necessary to employ the well-known compilation scheme proposed by Litinski~\cite{litinski2019game}.
\red{Under this compilation,} as shown in Fig.~\ref{fig_earlyFTQC_circuit}(b), all non-Clifford Pauli rotations are conjugated toward the beginning of the circuit, while Clifford operations are conjugated in such a way that they can be merged with measurement.
In this case, Clifford operations are not directly implemented on the logical circuits, and thus, noise on the Clifford operations does not need to be considered.

Conversely, when the number of magic state factories is sufficient, \red{we may prepare several magic states simultaneously, which we call multi-thread supply of magic states.
In this case,} such compilation is unnecessary.
Consequently, as depicted in Fig.~\ref{fig_earlyFTQC_circuit}(c), Clifford operations are directly executed via actual gate operations such as lattice surgery or gate teleportation.
In this scenario, both Clifford and non-Clifford operations are affected by logical errors.
Nevertheless, we assume that errors in non-Clifford operations dominate, allowing us to neglect errors in Clifford layers.
This assumption is particularly valid for Pauli rotation gates, where the gate count under Clifford+T synthesis may reach several tens to hundreds~\cite{ross2014optimal, yoshioka2024error}.
Considering that the T gate is also affected by distillation error, the logical error on Clifford gates is expected to be smaller by a factor of hundreds to thousands compared to Pauli rotation gates (see \red{Supplementary Note~1} for details).

Hereafter, our target for twirling is the noise affecting non-Clifford gates in both scenarios.
\red{Since one of the most promising applications in the early FTQC regime is Hamiltonian simulation or phase estimation based on Trotter decomposition~\cite{childs2018toward, campbell2021early, toshio2024practical},} we consider the case where the non-Clifford layer $\mathcal{U}_l$ is a Pauli-Z rotation gate $R_z(\theta) = e^{i\theta \mathrm{Z}}$ applied to the first qubit, formulated as $\mathcal{U}_l(\cdot) = \mathcal{U}(\cdot) = U\cdot U^\dag$ with $U = e^{i\theta \mathrm{Z}\otimes \mathrm{I}^{\otimes n-1}}$.
Our main focus is the noise $\mathcal{N}$ affecting the non-Clifford operation $\mathcal{U}$ as $\mathcal{N}\circ\mathcal{U}$.
We note that $\mathcal{U}(\cdot)$ can be transformed into general Pauli rotation gates by appropriately selecting Clifford operations $\mathcal{C}_l$.
Furthermore, our discussion can be generalized to arbitrary non-Clifford unitaries, \red{such as T gate and Toffoli gate (see \red{Supplementary Note~3} for details)}.

Notably, the noise $\mathcal{N}$ following the non-Clifford operation $\mathcal{U}$ tends to be Pauli noise. This is because noise affecting non-Clifford gates belonging to the third level of the Clifford hierarchy~\cite{gottesman1999quantum}, such as the T gate or the Toffoli gate, can be Pauli-twirled into Pauli noise~\cite{cai2023quantum}. Furthermore, algorithmic errors affecting synthesized non-Clifford gates can also be transformed into Pauli noise via randomized compiling~\cite{yoshioka2024error}. Therefore, we assume that the logical noise affecting the non-Clifford layer $\mathcal{U}$ is Pauli noise, expressed as
\begin{equation}
    \label{eq_pauli_noise}
    \mathcal{N} := (1-p_{\mathrm{err}})\mathcal{I} + \red{\sum_{i=1}^{4^n-1}} p_i\mathcal{E}_{P_i},    
\end{equation}
where Pauli error $\mathcal{E}_{P_i}(\cdot) := P_i\cdot P_i$ occurs with probability $p_i$, and $p_{\mathrm{err}} := \sum_i p_i$ represents the total error probability.
\red{We guide the readers to \red{Supplementary Note~4} for the case of general noise channel.}

Although the accumulation of noise $\mathcal{N}$ can degrade the logical quantum state, QEM techniques allow us to estimate the expectation value of some observables. A major drawback of QEM is the sampling overhead—defined as the multiplicative factor in the number of circuit executions required to restore the ideal expectation value—which increases exponentially with the number of noisy layers $L$~\cite{tsubouchi2023universal, takagi2023universal, quek2022exponentially}. However, if the total error rate $p_{\mathrm{tot}} := p_{\mathrm{err}}L$ remains constant, error mitigation can be achieved with a constant sampling overhead. In the early FTQC regime, the error probability $p_{\mathrm{err}}$ is expected to be lower than in the noisy intermediate-scale quantum (NISQ) regime, so we can apply QEM on deeper quantum circuits.

In particular, if the noise $\mathcal{N}$ can be transformed into the global white noise defined as
\begin{equation}
    \label{eq_white_noise}
    \mathcal{N}_{\mathrm{wn},p_{\mathrm{err}}} := (1-p_{\mathrm{err}})\mathcal{I} + p_{\mathrm{err}}\mathbb{E}_{P\in\mathcal{P}_n-\qty{I}^{\otimes n}}[\mathcal{E}_{P}],
\end{equation}
error mitigation can be achieved by simply rescaling the noisy expectation value \red{as $e^{p_{\mathrm{tot}}} \ev*{O}_{\mathrm{noisy}}$.}
Here, $\mathcal{P}_n:= \qty{\mathrm{I},\mathrm{X},\mathrm{Y},\mathrm{Z}}^{\otimes n}$ denotes the set of $n$-qubit Pauli operators, $\mathbb{E}$ represents the uniform average, \red{and $\ev*{O}_{\mathrm{noisy}}$ is the noisy expectation value obtained from the noisy logical circuit.}
The sampling overhead for rescaling \red{is $e^{2p_{\mathrm{tot}}}$, which is the square of the rescaling coefficient $e^{p_{\mathrm{tot}}}$.
This} not only represents a quadratic improvement over the previous probabilistic error cancellation approach, which scales as $e^{4p_{\mathrm{tot}}}$~\cite{piveteau2021error, lostaglio2021error, suzuki2022quantum}, but also achieves the theoretical lower bound on the sampling overhead~\cite{tsubouchi2023universal} (see \red{Supplementary Note~2} for details). Thus, in this work, we aim to convert the noise $\mathcal{N}$ into global white noise $\mathcal{N}_{\mathrm{wn},p_{\mathrm{err}}}$ for the cost-optimal QEM.

\subsection*{Symmetric Clifford twirling}
One naive way of converting noise to global white noise is to perform Clifford twirling~\cite{emerson2005scalable, dankert2009exact}.
Clifford twirling scrambles the noise $\mathcal{N}$ into the global white noise $\mathcal{N}_{\mathrm{wn},p_{\mathrm{err}}}$ by applying random Clifford unitary $D\in \mathcal{G}_n$ and its conjugation $D^\dagger$ before and after the noise $\mathcal{N}$ as
\begin{equation}
    \label{eq_twirling}
    \mathscr{T}(\mathcal{N})
        := \mathbb{E}_{D\in\mathcal{G}_{n}}[\mathcal{D}^\dagger\circ\mathcal{N}\circ\mathcal{D}]
        = \mathcal{N}_{\mathrm{wn},p_{\mathrm{err}}}.
\end{equation}
Here, $\mathcal{G}_n$ represents the $n$-qubit Clifford group, $\mathcal{D}(\cdot) := D\cdot D^\dagger$, and $\mathscr{T}$ denotes the superchannel representing Clifford twirling.
This operation is, however, not considered as a practical option in the community, since the noise is inseparable from the target non-Clifford unitary $\mathcal{U}(\cdot)=U\cdot U^\dag$.
In order to insert $D$ before the noise $\mathcal{N}$, one must insert $U^\dag D U$ before the noisy non-Clifford layer $\mathcal{N}\circ\mathcal{U}$, which may introduce additional errors if $U^\dag D U$ is an intricate non-Clifford unitary.

\begin{table*}[t]
  \caption{Distance $v$ between the Pauli noise $\mathcal{N}$ defined as in Eq.~\eqref{eq_singleq_pauli_noise} and global white noise $\mathcal{N}_{\mathrm{wn},p_{\mathrm{err}}}$ with the same error rate $p_{\mathrm{err}}$.
  This table presents values of $v$ defined in Eq.~\eqref{eq_2-norm} for three scenarios: the original noise, after applying symmetric Clifford twirling, and after applying $k$-sparse symmetric Clifford twirling. Here, we only consider the leading order of $n$.}
  \label{tbl_}
  \centering
  \begin{tabular}{c||c|c}
      Noise model &  With Pauli-Z component ($p_z \neq 0$) & Without Pauli-Z component ($p_z = 0$) \\\hline\hline
      Original noise & $\sqrt{p_x^2+p_y^2+p_z^2}/\red{p_{\mathrm{err}}}$ & $\sqrt{p_x^2+p_y^2}/\red{p_{\mathrm{err}}}$ \\ 
      Symmetric Clifford twirling  & $p_z/p_{\mathrm{err}}$ & $O(2^{-n})$\\ 
      $k$-sparse symmetric Clifford twirling & $p_z/p_{\mathrm{err}}$ & $O(n^{-(k-1)/2})$\\
  \end{tabular}
\end{table*}

\begin{figure}[t]
    \begin{center}
        \includegraphics[width=0.99\linewidth]{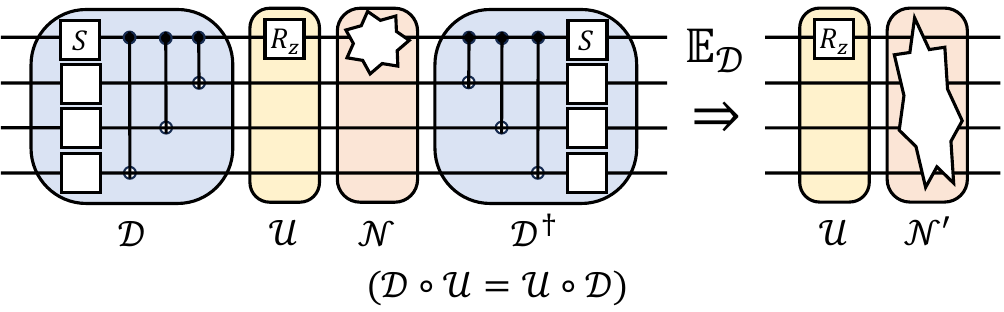}
        \caption{Conceptual diagram of symmetric Clifford twirling. By randomly sampling Clifford unitary $\mathcal{D}(\cdot)=D\cdot D^\dagger$ that commutes with non-Clifford layer $\mathcal{U}(\cdot)=U\cdot U^\dagger$, we can scramble the noise $\mathcal{N}$ without affecting $\mathcal{U}$.
        }
        \label{fig_SCT}
    \end{center}
\end{figure}

One feasible alternative is to consider Clifford unitaries $D$ that commute with $U$, since $U^\dag D U = D$ becomes a Clifford unitary.
To characterize Clifford unitaries commuting with $U$, let us define a Pauli subgroup
\begin{equation}
    \label{eq_Q_U}
    \mathcal{Q}_{U} :=\ev*{\qty{P\in\mathcal{P}_n\;|\;\mathrm{tr}[PU]\neq0}},
\end{equation}
where $\ev*{\cdot}$ represents the group generated by the elements within the bracket.
Additionally, let us define the $\mathcal{Q}_{U}$-symmetric Clifford group as:
\begin{equation}
    \mathcal{G}_{n,\mathcal{Q}_{U}} := \qty{C\in\mathcal{G}_{n} \;|\; \forall P \in \mathcal{Q}_{U}, \;[C,P]=0},
\end{equation}
where its complete and unique construction method using simple quantum gates is given in Ref.~\cite{mitsuhashi2023clifford}.
From the definition, $D\in\mathcal{G}_{n,\mathcal{Q}_{U}}$ commutes with the non-Clifford operator $U$, so we can twirl the noise layer  $\mathcal{N}$ using $D\in\mathcal{G}_{n,\mathcal{Q}_{U}}$ with negligible errors.
We term such twirling as {\it symmetric Clifford twirling}, whose effect is represented using a superchannel defined as:
\begin{equation}
    \mathscr{T}_{Q_U}(\mathcal{N})
    :=\mathbb{E}_{D\in\mathcal{G}_{n, \mathcal{Q}_U}}[\mathcal{D}^\dagger\circ\mathcal{N}\circ\mathcal{D}].
\end{equation}

For the sake of clarity, let us consider the scenario where the non-Clifford layer $\mathcal{U}$ consists of  Pauli-Z rotation gate $R_z(\theta) = e^{i\theta \mathrm{Z}}$ applied to the first qubit as $U=e^{i\theta \mathrm{Z}\otimes \mathrm{I}^{\otimes n-1}}$ (see Fig.~\ref{fig_SCT}).
We assume that Pauli-X, Y, and Z noises affect the Pauli-Z rotation gate with probability $p_x$, $p_y$, and $p_z$, which is characterized as a noise channel 
\begin{equation}
    \label{eq_singleq_pauli_noise}
    \begin{aligned}
        \mathcal{N} &= (1-p_{\mathrm{err}})\mathcal{I} \\
        &+p_x\mathcal{E}_{\mathrm{X}\otimes \mathrm{I}^{\otimes n-1}} + p_y\mathcal{E}_{\mathrm{Y}\otimes \mathrm{I}^{\otimes n-1}} +p_z\mathcal{E}_{\mathrm{Z}\otimes \mathrm{I}^{\otimes n-1}}.
    \end{aligned}
\end{equation}
In this particular case, $\mathcal{Q}_U$ simplifies to $\mathcal{Q}_{U} = \qty{\mathrm{I},\mathrm{Z}}\otimes \qty{\mathrm{I}}^{\otimes n-1}$, and we can express the effect of symmetric Clifford twirling to the Pauli noise as presented in the following theorem.

\begin{thm}[Symmetric Clifford twirling of single-qubit Pauli channel]
    \label{thm_1}
    Let $\mathcal{E}_{P\otimes \mathrm{I}^{\otimes{n-1}}}(\cdot) = (P\otimes\mathrm{I}^{\otimes n-1})\cdot (P\otimes\mathrm{I}^{\otimes n-1})$ be a single-qubit Pauli channel with $P = \mathrm{X}, \mathrm{Y}, \mathrm{Z}$ and $\mathcal{Q}_{U} = \qty{\mathrm{I},\mathrm{Z}}\otimes \qty{\mathrm{I}}^{\otimes n-1}$.
    Then, by applying symmetric Clifford twirling to the Pauli channel as $\mathscr{T}_{\mathcal{Q}_U}(\mathcal{E}_{P\otimes \mathrm{I}^{\otimes{n-1}}})
    =\mathbb{E}_{D\in\mathcal{G}_{n, \mathcal{Q}_U}}[\mathcal{D}^\dagger\circ\mathcal{E}_{P\otimes \mathrm{I}^{\otimes{n-1}}}\circ\mathcal{D}]$, we can scramble the Pauli-X and Y channels as
    \begin{equation}
        \mathscr{T}_{\mathcal{Q}_{U}} (\mathcal{E}_{P\otimes \mathrm{I}^{\otimes{n-1}}})
        = \underset{\substack{Q_1\in\qty{\mathrm{X}, \mathrm{Y}} \\ Q_2\in\mathcal{P}_{n-1}}}{\mathbb{E}}\qty[\mathcal{E}_{Q_1\otimes Q_2}]
    \end{equation}
    for $P = \mathrm{X}, \mathrm{Y}$, while the Pauli-Z channel cannot be scrambled through the symmetric Clifford twirling:
    \begin{equation}
        \mathscr{T}_{\mathcal{Q}_{U}}(\mathcal{E}_{\mathrm{Z}\otimes \mathrm{I}^{\otimes{n-1}}})
        =\mathcal{E}_{\mathrm{Z}\otimes \mathrm{I}^{\otimes{n-1}}}.
    \end{equation}
\end{thm}
We generalize Theorem~\ref{thm_1} to arbitrary non-Clifford unitaries and Pauli operators in \red{Supplementary Note~3}.

Let us evaluate how effectively symmetric Clifford twirling scrambles noise into global white noise.
As a performance metric for twirling, we introduce a measure to evaluate the proximity of Pauli noise to global white noise.
Given the Pauli noise in Eq.~\eqref{eq_pauli_noise} and the global white noise in Eq.~\eqref{eq_white_noise} with the same error rate $p_{\mathrm{err}}$, we define the 2-norm $v$ of the normalized error probabilities as
\begin{equation}
    \label{eq_2-norm}
    v := \sqrt{\red{\sum_{i=1}^{4^n-1}}\qty(\frac{p_i}{p_{\mathrm{err}}} - \frac{1}{4^n-1})^2} = \sqrt{\red{\sum_{i=1}^{4^n-1}}\qty(\frac{p_i}{p_{\mathrm{err}}})^2 - \frac{1}{4^n-1}}.
\end{equation}
This distance measure $v$ is valuable as it helps bound the bias between the ideal and the rescaled noisy expectation values on average (see \red{Supplementary Note~6} for details). By analyzing how $v$ changes through noise conversion, we can assess its efficacy in QEM. We note that the commonly used diamond norm, normalized by error probability, corresponds to the 1-norm of normalized error probabilities~\cite{magesan2012characterizing}:
\begin{equation}
    \norm{\mathcal{N} - \mathcal{N}_{\mathrm{wn}, p_{\mathrm{err}}}}_{\Diamond}/p_{\mathrm{err}} = \red{\sum_{i=1}^{4^n-1}} \abs{\frac{p_i}{p_{\mathrm{err}}} - \frac{1}{4^n-1}}.
\end{equation}

For single-qubit Pauli noise represented as Eq.~\eqref{eq_singleq_pauli_noise}, the dominant term in the distance $v$ defined in Eq.~\eqref{eq_2-norm} is $\sqrt{p_x^2+p_y^2+p_z^2}/p_{\mathrm{err}}$.
When symmetric Clifford twirling is applied to this noise, Theorem~\ref{thm_1} indicates that Pauli-Z noise remains unscrambled, while Pauli-X and Y noise is well dispersed among other qubits.
This occurs because Pauli-Z noise commutes with symmetric Clifford operations and remains unchanged, whereas Pauli-X and Y noise propagates to other qubits through conjugation with symmetric Clifford operations.
Indeed, the twirled noise $\mathscr{T}_{\mathcal{Q}_{U}}(\mathcal{N})$ is given by
\begin{equation}
    \label{eq_global_twirled}
    \begin{aligned}
        &~~(1-p_{\mathrm{err}})\mathcal{I} \\
        &+ p_{\mathrm{err}}\qty(\frac{p_x+p_y}{p_{\mathrm{err}}}\underset{\substack{Q_1\in\qty{\mathrm{X}, \mathrm{Y}} \\\ Q_2\in\mathcal{P}_{n-1}}}{\mathbb{E}}\qty[\mathcal{E}_{Q_1\otimes Q_2}] + \frac{p_z}{p_{\mathrm{err}}}\mathcal{E}_{\mathrm{Z}\otimes \mathrm{I}^{\otimes n-1}}).
    \end{aligned}
\end{equation}

When the noise is biased with no Pauli-Z component, i.e., $p_z = 0$, the distance $v$ between $\mathscr{T}_{\mathcal{Q}_{U}}(\mathcal{N})$ and $\mathcal{N}_{\mathrm{wn},p_{\mathrm{err}}}$ is calculated as $v = O(2^{-n})$ up to the leading order.
Therefore, the twirled noise $\mathscr{T}_{\mathcal{Q}_{U}}(\mathcal{N})$ converges exponentially to global white noise.
However, when the noise includes a Pauli-Z component, i.e., $p_z \neq 0$, the leading term of the distance is $p_z/p_{\mathrm{err}}$ (see Table~\ref{tbl_}).
While there is a constant decrease in the distance, symmetric Clifford twirling does not induce exponential decay in this case.
This motivates us to develop methods for implementing the $R_z(\theta)$ gate such that the dominant error is Pauli-X or Y error~\cite{yoshioka2024error}.
Alternatively, we may focus on mitigating Pauli-Z noise using probabilistic error cancellation while employing symmetric Clifford twirling to address residual Pauli-X and Y noise.

It is noteworthy that an $n$-qubit Clifford unitary $D \in \mathcal{G}_{n,\mathcal{Q}}$ can be implemented with negligible effort in terms of (i) measurement shots, (ii) classical computational cost for randomly selecting gates from the symmetric Clifford group $\mathcal{G}_{n,\mathcal{Q}}$, and (iii) additional error.
The first point follows from the general characteristics of randomized compiling~\cite{wallman2016noise}.
Regarding the second point, although the group size scales superexponentially, a polynomial-time sampling algorithm exists~\cite{mitsuhashi2023clifford}.
Moreover, to scramble single-qubit Pauli noise Eq.~\eqref{eq_singleq_pauli_noise} into Eq.~\eqref{eq_global_twirled}, it is not necessary to sample from all elements of $\mathcal{G}_{n, \mathcal{Q}_U}$.
Instead, as shown in Fig.~\ref{fig_SCT}, noise can be probabilistically propagated to $n-1$ idling qubits, followed by local twirling using single-qubit Clifford operations.
This process is achieved by constructing $D\in\mathcal{G}_{n, \mathcal{Q}_U}$ as follows:
\begin{enumerate}
    \item For $i = 2,\dots,n$, add $i$ to the set $\mathbf{Target}$ with probability $3/4$.
    \item Add a multi-qubit CNOT gate, with the first qubit as the control qubit and the qubits in $\mathbf{Target}$ as target qubits.
    \item For qubits in the set $\mathbf{Target}$, apply a randomly and uniformly sampled single-qubit Clifford gate.
    \item With probability $1/2$, apply an S gate to the first qubit.
\end{enumerate}
\red{We note that S gate is defined as $\mathrm{diag}(1,i)$ in the computational basis.}

Regarding the third point, we examine additional error effects for the scenarios depicted in Fig.~\ref{fig_earlyFTQC_circuit}(b) and Fig.~\ref{fig_earlyFTQC_circuit}(c).
Consider inserting additional random Clifford gates $D_l$ and $D^\dag_l$ to twirl the $l$-th non-Clifford gate in the original circuit of Fig.~\ref{fig_earlyFTQC_circuit}(a).
In the setup of Fig.~\ref{fig_earlyFTQC_circuit}(b), before conjugating non-Clifford Pauli rotation gates toward the beginning, we conjugate $D_l$ toward state preparation and $D_l^\dag$ toward measurement.
Then, we merge $D_l$ with state preparation and $D_l^\dag$ with measurements, followed by conjugation of non-Clifford Pauli rotation gates toward the beginning.
This produces a randomized compiled circuit whose output remains equivalent when noise is absent.
Thus, applying symmetric Clifford twirling to non-Clifford gates results in randomized compiling of the entire circuit.

Through randomized compiling, the weights of multi-qubit Pauli rotations and multi-Pauli measurements may change.
However, if the original circuit is sufficiently deep, e.g., depth $\Omega(n)$ for an $n$-qubit simulation, these weights are already on the order of \red{$\Omega(n)$}.
In this regime, the total error is not significantly affected by the addition of Clifford operations $D_l$.
Since the number and weights of non-Clifford operations and final multi-Pauli measurements remain unchanged, the additional overhead is limited to stabilizer state preparation, which is negligible compared to the full circuit (see \red{Supplementary Note~1} for details).

For the setup in Fig.~\ref{fig_earlyFTQC_circuit}(c), additional random Clifford gates $D_l$ are implemented directly via lattice surgery or gate teleportation.
Although this introduces additional logical errors, these errors are much smaller than those of Pauli rotation gates.
This is because controlled multi-qubit Pauli gates can be implemented with $O(1)$ depth via lattice surgery or gate teleportation, whereas Pauli rotation gates require several tens to hundreds of Clifford+T gates for synthesis~\cite{ross2014optimal, yoshioka2024error}.
Considering that T gates are also affected by distillation errors and that Clifford gate errors can be detected during gate teleportation~\cite{delfosse2024low}, errors on additional gates $D_l$ are expected to be hundreds to thousands of times smaller than those on the \red{non-Clifford} gate $U$ per qubit (see \red{Supplementary Note~1} for details).

\subsection*{$k$-Sparse Symmetric Clifford Twirling}

\begin{figure}[t]
    \begin{center}
        \includegraphics[width=0.9\linewidth]{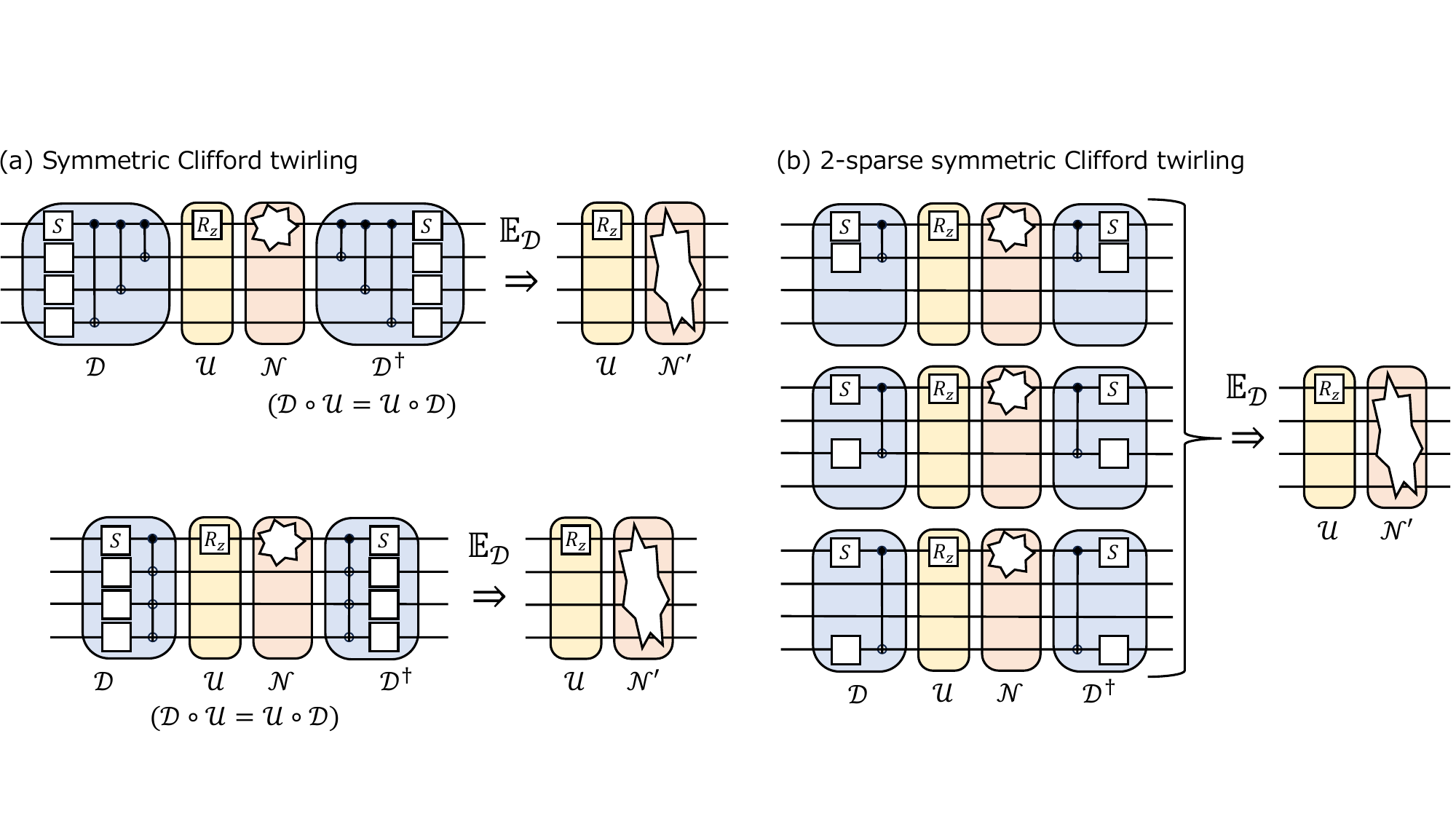}
        \caption{Schematic representation of 2-sparse symmetric Clifford twirling.
        While symmetric Clifford twirling propagates noise to all $n-1$ idling qubits, 2-sparse symmetric Clifford twirling propagates noise only to a single idling qubit, randomly selected for each circuit execution.}
        \label{fig_k-sparse}
    \end{center}
\end{figure}

In the previous section, we \red{argued} that logical errors per qubit on an additional random Clifford operation $D$ are small compared to the \red{non-Clifford} gate $U$ in the setup of Fig.~\ref{fig_earlyFTQC_circuit}(c).
Nevertheless, implementing the Clifford operation $D$ affects $O(n)$ qubits.
Even though such $D$ can be implemented with $O(1)$ depth, its logical error rate is expected to scale as $O(n)$.
Meanwhile, if we can build a simplified method with local Clifford gates that resembles the function of symmetric Clifford twirling, the logical error rate can be expected to be suppressed.

For this purpose, we consider limiting the sampled symmetric Clifford unitaries $D\in\mathcal{G}_{n, \mathcal{Q}_U}$ to at most $k$-qubit unitaries.
In other words, instead of propagating local noise to all $n$ qubits, we simplify the implementation of $D$ by restricting noise propagation to at most $k$ qubits.
As an example of such a strategy, we construct $D$ as follows:
\begin{enumerate}
    \item Sample $k'\in\qty{0,\dots,k-1}$ with probability $3^{k'}\tbinom{n-1}{k'}/\sum_{k''=0}^{k-1}3^{k''}\tbinom{n-1}{k''}$.
    \item Randomly and uniformly select $k'$ qubits from the idling $n-1$ qubits and add them to the set $\mathbf{Target}$.
    \item Apply a multi-qubit CNOT gate with the first qubit as the control qubit and the qubits in $\mathbf{Target}$ as target qubits.
    \item For qubits in $\mathbf{Target}$, apply a randomly and uniformly sampled single-qubit Clifford gate.
    \item With probability $1/2$, apply an S gate to the first qubit.
\end{enumerate}

Step 1 determines the number of qubits to which the noise propagates, where $\binom{n-1}{k'}$ is the binomial coefficient and $3^{k'}\binom{n-1}{k'}$ represents the number of $(n-1)$-qubit Pauli operators with Pauli weight $k'$.
The Pauli weight refers to the number of qubits on which a Pauli operator acts nontrivially.
Step 2 selects the $k'$ qubits that will receive the propagated noise, while step 3 propagates the noise to these qubits.
Finally, steps 4 and 5 locally scramble the noise.

Since the constructed $D\in\mathcal{G}_{n, \mathcal{Q}_U}$ affects at most $k$ qubits, we refer to twirling using such $D$ as \textit{$k$-sparse symmetric Clifford twirling} (see Fig.~\ref{fig_k-sparse} for $k=2$).
This method scrambles noise into the uniform average of Pauli noise with Pauli weights up to $k-1$.
Thus, the effect of $k$-sparse symmetric Clifford twirling for $P=\mathrm{X}, \mathrm{Y}$ can be mathematically described as:
\begin{equation}
    \begin{aligned}
        \mathscr{T}_{\mathcal{Q}_{U}}^{k\text{-sparse}} (\mathcal{E}_{P\otimes \mathrm{I}^{\otimes{n-1}}})
    &=\mathbb{E}_{D\in\mathcal{G}_{n, \mathcal{Q}_U}}^{k\text{-sparse}}[\mathcal{D}^\dagger\circ\mathcal{E}_{P\otimes \mathrm{I}^{\otimes{n-1}}}\circ\mathcal{D}]\\
    &= \underset{\substack{Q_1\in\qty{\mathrm{X}, \mathrm{Y}} \\ Q_2\in\mathcal{P}_{n-1}, w(Q_2)\leq k-1}}{\mathbb{E}}\qty[\mathcal{E}_{Q_1\otimes Q_2}],
    \end{aligned}
\end{equation}
where $\mathbb{E}_{D\in\mathcal{G}_{n, \mathcal{Q}_U}}^{k\text{-sparse}}$ represents the average over the above construction of $D$, and $w(Q_2)$ denotes the Pauli weight of $Q_2$.
Therefore, applying $k$-sparse symmetric Clifford twirling to single-qubit Pauli noise in Eq.~\eqref{eq_singleq_pauli_noise} results in a twirled noise $\mathscr{T}_{\mathcal{Q}_{U}}^{k\text{-sparse}}(\mathcal{N})$ given by
\begin{equation}
    \begin{aligned}
        &~~(1-p_{\mathrm{err}})\mathcal{I} \\
        &+ p_{\mathrm{err}}\qty(\frac{p_x+p_y}{p_{\mathrm{err}}}\underset{\substack{Q_1\in\qty{\mathrm{X}, \mathrm{Y}} \\ Q_2\in\mathcal{P}_{n-1} \\ w(Q_2)\leq k-1}}{\mathbb{E}}\qty[\mathcal{E}_{Q_1\otimes Q_2}] + \frac{p_z}{p_{\mathrm{err}}}\mathcal{E}_{\mathrm{Z}\otimes \mathrm{I}^{\otimes n-1}}).
    \end{aligned}
\end{equation}

When the noise is biased with no Pauli-Z component, i.e., $p_z = 0$, the resulting noise channel is a uniform average over $2\sum_{k''=0}^{k-1}3^{k''}\tbinom{n-1}{k''} \sim 2(3n)^{k-1}/(k-1)!$ types of Pauli noise.
Thus, the distance $v$ between $\mathscr{T}_{\mathcal{Q}_{U}}(\mathcal{N})$ and $\mathcal{N}_{\mathrm{wn},p_{\mathrm{err}}}$ is calculated as $v=O(n^{-(k-1)/2})$ up to the leading order.
This implies that although exponential convergence is unattainable unlike full symmetric Clifford twirling, Pauli-X and Y noise can be scrambled polynomially close to global white noise.
Moreover, when the noise includes a Pauli-Z component, i.e., $p_z \neq 0$, the leading term of the distance is $p_z/p_{\mathrm{err}}$, which is the same as the full twirling (see Table~\ref{tbl_}).
These results demonstrate the effectiveness of the simplified $k$-sparse twirling approach, which only requires $k$-qubit symmetric Clifford operators.

\subsection*{Numerical Analysis}
\begin{figure*}[t]
    \begin{center}
        \resizebox{0.99\hsize}{!}{\includegraphics{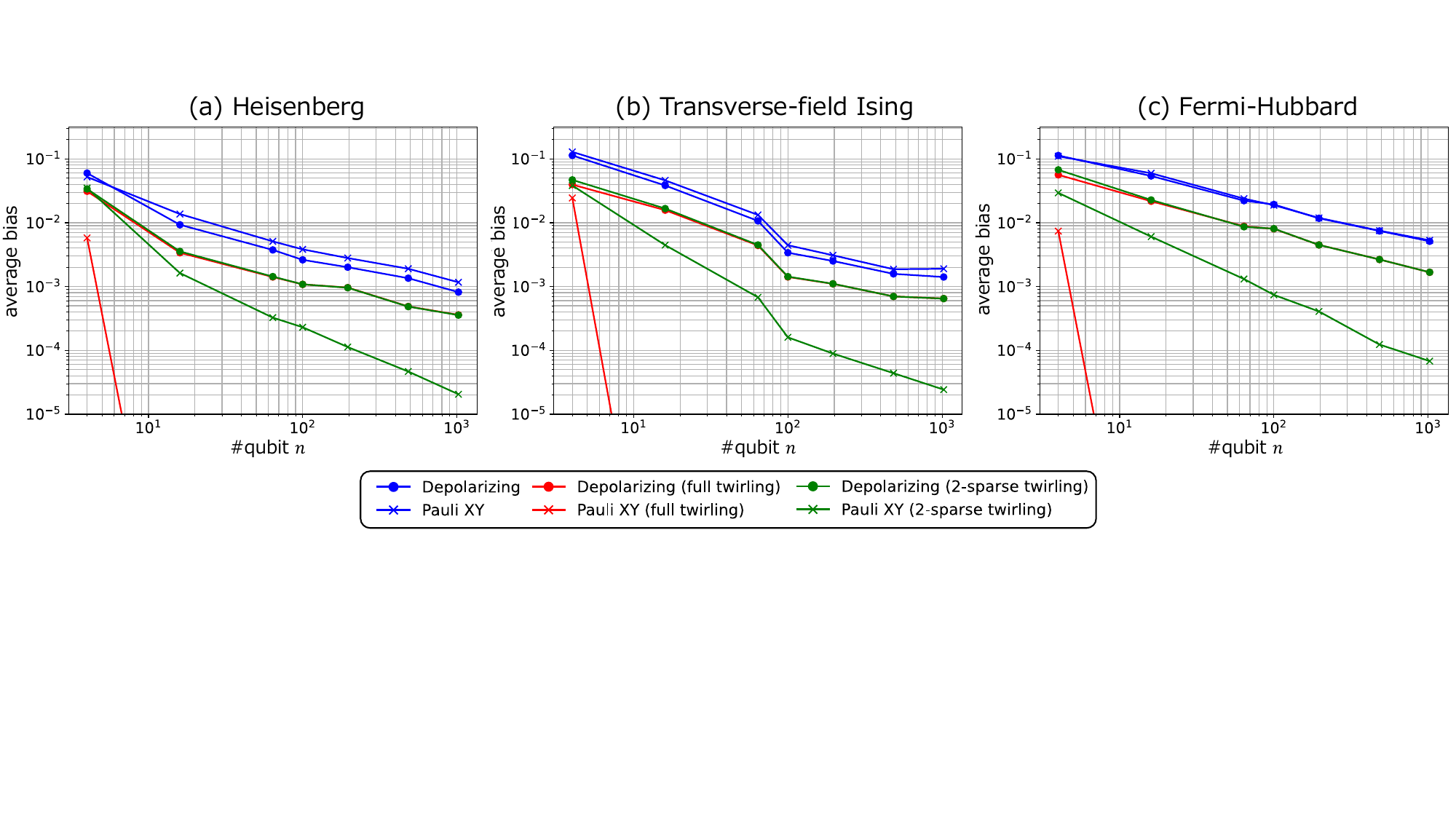}}
        \caption{Qubit count dependence of the performance of symmetric Clifford twirling for Trotterized Hamiltonian simulation of (a): 2D Heisenberg model, (b): 2D transverse-field Ising model, and (c): 2D Fermi-Hubbard model. These \red{graphs} represent the average bias over random Pauli operators $P\in \mathcal{P}_n$ defined as $\abs{R2^{-n}\abs{\mathrm{tr}[\mathcal{N}_{\mathrm{eff}}(P)P]} - 1}$, where $\mathcal{N}_{\mathrm{eff}}$ is the effective noise channel of the entire logical circuit and $R$ is the optimal rescaling coefficient defined in \red{Supplementary Note~6}.
        We fix the total error rate as $p_{\mathrm{tot}} := p_{\mathrm{err}}L = 1$ and set the Trotter step size as 100. The circle and the x dots respectively represent the results for depolarizing noise and Pauli-X and Y noise, while the \red{blue, red, and green} lines represent the results without symmetric Clifford twirling, with symmetric Clifford twirling, and with 2-sparse symmetric Clifford twirling. Here we represent the result for Clifford simulation, where all the angles of Pauli rotation are taken as $\theta = \pi/4$.
        }
        \label{fig_numerics}
    \end{center}
\end{figure*}
\begin{figure*}[t]
    \begin{center}
        \resizebox{0.99\hsize}{!}{\includegraphics{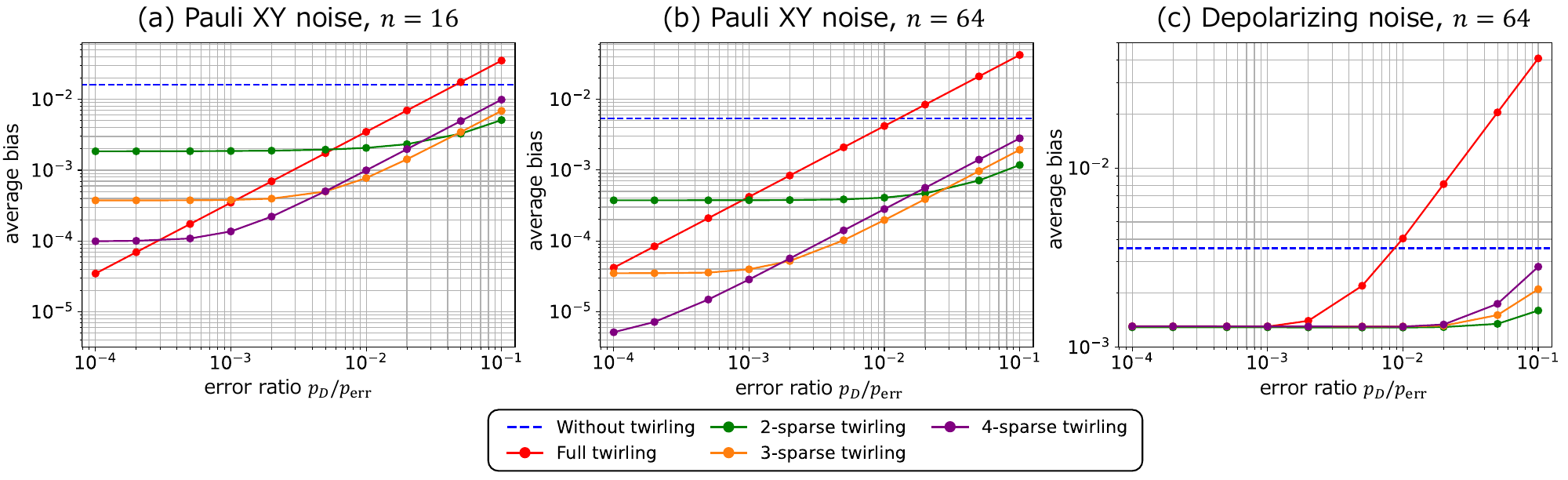}}
        \caption{Performance of symmetric Clifford twirling in the presence of noise on the additional symmetric Clifford operations for twirling.
        We assume that twirling operations are subject to local depolarizing noise with error rate $p_D$.
        These \red{graphs} represent the average bias over random Pauli operators $P\in \mathcal{P}_n$ defined as $\abs{R2^{-n}\abs{\mathrm{tr}[\mathcal{N}_{\mathrm{eff}}(P)P]} - 1}$, where $\mathcal{N}_{\mathrm{eff}}$ is the effective noise channel of the entire logical circuit and $R$ is the optimal rescaling coefficient defined in \red{Supplementary Note~6}.
        The circuit considered is the $n$-qubit Trotterized Hamiltonian simulation circuit of the 2D Heisenberg model with Trotter step size 100 and total error rate $p_{\mathrm{tot}} := p_{\mathrm{err}}L = 1$.
        Panels (a), (b), and (c) represent results for $n=16$ qubits with a Pauli-Z rotation gate affected by Pauli-X and Y noise, $n=64$ qubits with Pauli-X and Y noise, and $n=64$ qubits with depolarizing noise, respectively.
        Here, we present results for Clifford simulation, where all the angles of Pauli rotations are set to $\theta = \pi/4$.}
        \label{fig_numerics_noisy}
    \end{center}
\end{figure*}
We demonstrate the impact of ($k$-sparse) symmetric Clifford twirling in mitigating errors via the dynamics simulation circuit of the first-order Suzuki-Trotter decomposition for the 2D Heisenberg model, the 2D transverse-field Ising model, and the 2D Fermi-Hubbard model under open boundary conditions.
In this setup, we assume that the noisy non-Clifford layer is given as $\mathcal{N} \circ \mathcal{U}$ instead of the noiseless $\mathcal{U}(\cdot) = U\cdot U^\dag$, where $U = e^{i\theta \mathrm{Z}\otimes \mathrm{I}^{\otimes n-1}}$ represents a Pauli-Z rotation gate.
We sandwich this non-Clifford layer with Clifford operations to implement the target Pauli rotation.
As for the noise, we assume that $\mathcal{N}$ is a single-qubit Pauli error as in Eq.~\eqref{eq_singleq_pauli_noise}.
We specifically consider Pauli-X and Y noise with $p_x = p_y = p_{\mathrm{err}}/2$, or depolarizing noise with $p_x = p_y = p_z = p_{\mathrm{err}}/3$.
We set the total error rate to $p_{\mathrm{tot}} = 1$, resulting in a constant sampling overhead of $e^2 \approx 7$.

To demonstrate the effectiveness of symmetric Clifford twirling in large-scale quantum circuits, we replace the $R_z(\theta)$ gate with \red{$R_z(\pi/4)$ gate, corresponding to $S$ gate up to phase}, and perform Clifford simulations, as shown in Fig.~\ref{fig_numerics}.
We guide the readers to \red{Supplementary Note~6} for the non-Clifford state simulation, where the similarity between Clifford and non-Clifford simulation is discussed.
Figure~\ref{fig_numerics} represents the average bias over random Pauli operators $P\in \mathcal{P}_n$ defined as $\abs{R2^{-n}\abs{\mathrm{tr}[\mathcal{N}_{\mathrm{eff}}(P)P]} - 1}$, where $\mathcal{N}_{\mathrm{eff}}$ is the effective noise channel of the entire logical circuit and $R$ is the optimal rescaling coefficient (see \red{Supplementary Note~6} for details).
We observe that symmetric Clifford twirling significantly reduces the bias for Pauli-X and Y noise, corroborating Theorem~\ref{thm_1}, which states that such noise can be twirled into an exponentially close approximation of global white noise.
This result motivates us to develop methods for synthesizing the $R_z(\theta)$ gate such that the dominant error consists of Pauli-X or Y noise.
Such gate compilation is achievable probabilistically in certain scenarios~\cite{yoshioka2024error}, although the general feasibility of this approach remains an open question.

Additionally, we observe that the average bias between the ideal and error-mitigated expectation values decreases approximately as $1/\sqrt{n}$ as the qubit count increases, even without applying symmetric Clifford twirling.
This effect is consistent with the white-noise approximation~\cite{dalzell2021random}, where the effective noise of the quantum circuit is scrambled into global white noise on average.
In \red{Supplementary Note~6}, we show that when each Clifford layer $C_l$ is randomly sampled, the bias scales as
\begin{equation}
    \abs{R2^{-n}\abs{\mathrm{tr}[\mathcal{N}_{\mathrm{eff}}(P)P] - 1}} \sim \frac{vp_{\mathrm{tot}}}{\sqrt{n}}.
\end{equation}
Our numerical results demonstrate that this scaling also holds for Trotterized Hamiltonian simulation circuits, which is not random but highly structured circuit.

As observed in our numerical simulations, symmetric Clifford twirling accelerates noise scrambling in the white-noise approximation.
This remains true even for depolarizing noise containing Pauli-Z components, as $v$ decreases from $\sqrt{p_x^2+p_y^2+p_z^2}/p_{\mathrm{err}} = 1/\sqrt{3}$ to $p_z/p_{\mathrm{err}} = 1/3$.
Moreover, noise scrambling can be further accelerated by applying sparse twirling techniques such as $2$-sparse twirling, which uses only a single CNOT gate.
In the absence of Pauli-Z error, $2$-sparse twirling reduces the distance from $v=O(1)$ to $v = O(1/\sqrt{n})$, improving the bias scaling from $1/\sqrt{n}$ to $1/n$.
We emphasize that the difference in the performance of the full twirling and the sparse twirling is negligible when the Pauli-Z noise remains untwirled, as we can see from the results of depolarizing noise.
This is because the leading term of the distance $v$ remains the same for both methods.

Next, we evaluate the efficacy of our methods in the presence of noise on the additional symmetric Clifford operation $D$ used for twirling.
We assume that qubits on which $D$ acts nontrivially are affected by local depolarizing noise with an error rate of $p_D$.
We analyze the average bias for varying ratios of the Clifford operation error rate $p_D$ to the Pauli rotation gate error rate $p_{\mathrm{err}}$, as shown in Fig.~\ref{fig_numerics_noisy}.

As we discussed in previous sections, we expect the error ratio $p_D/p_{\mathrm{err}}$ to be as small as $10^{-2}$ to $10^{-3}$.
In these regimes, we observe that even full symmetric Clifford twirling reduces the bias.
However, the performance of full twirling deteriorates as the qubit count $n$ increases, since the effect of twirling errors becomes more significant.

On the other hand, for $k$-sparse twirling, the noise affects only up to $k$ qubits per twirling operation.
As a result, the impact of the noise does not depend on the qubit count $n$, making its performance superior to full twirling in such cases.
Additionally, we observe a reduction in bias as $n$ increases, similar to the noiseless case.
Because of this, we find that the optimal sparsity $k$ varies depending on the qubit count $n$ and the error ratio $p_D/p_{\mathrm{err}}$.
%Thus, selecting the optimal $k$ based on the qubit count of the circuit and the error ratio is beneficial.
For depolarizing noise, however, we find that $2$-sparse twirling consistently performs best.
This is because the leading term of the distance $v$ remains independent of sparsity, making it preferable to choose a twirling method that introduces minimal additional noise.

\red{\section*{Discussion}}
In this work, we have introduced symmetric Clifford twirling, which converts local noise into noise resembling global white noise.
Additionally, we show that the effective noise of deep logical circuits, \red{even if they are} highly structured ones, can be regarded as global white noise on average.
These findings pave the way for mitigating logical errors in non-Clifford operations with minimal sampling overhead, or implementing early FTQC algorithms robust to global white noise~\cite{o2019quantum, katabarwa2024early, ding2023robust, dutkiewicz2024error}.

We foresee various future directions, highlighting the top two.
The first is to develop a way to perform non-Clifford unitary with a noise that can be scrambled through symmetric Clifford twirling.
For example, the common noise model of T gates executed via gate teleportation is known to be Pauli-Z noise~\cite{piveteau2021error, lostaglio2021error}, which cannot be twirled.
The development of a novel approach to perform T gates with Pauli X or Y errors represents a crucial avenue for future research.
\red{One possible way is to use code switching to a code where T gates can be implemented transversely~\cite{anderson2014fault, beverland2021cost}, where Pauli-Z error may not be a dominant component in the noise.}

The second is to develop a further application of symmetric Clifford twirling.
While our focus in this work has been on cost-optimal QEM, conventional Clifford twirling is also employed in various contexts ranging from fidelity estimation~\cite{emerson2005scalable, dankert2009exact} to the analysis of information loss in black holes~\cite{hayden2007black}.
Investigating the practical utility of symmetric Clifford twirling for such important tasks in quantum information theory, high energy physics, and many-body physics is left as an interesting future work.

\red{\section*{Data availability}
All study data are included in this article and \red{Supplementary Note}.
\section*{Code availability}
Codes are available upon request.}

\section*{Acknowledgement}
The authors wish to thank Gregory Boyd, Anthony Chen, Suguru Endo, Soumik Ghosh, Christophe Piveteau, Yihui Quek, Takahiro Sagawa, Yasunari Suzuki, Kristan Temme, Ewout van den Berg, and Yat Wong for fruitful discussions.
K.T. and Y.M. are supported by World-leading Innovative Graduate Study Program for Materials Research, Information, and Technology (MERIT-WINGS) of the University of Tokyo.
K.T. is also supported by JSPS KAKENHI Grant No. JP24KJ0857 and JST BOOST Grant Number JPMJBS2418.
Y.M. is also supported by JSPS KAKENHI Grant No. JP23KJ0421.
N.Y. wishes to thank 
JST PRESTO No. JPMJPR2119, % sakigake
JST Grant Number JPMJPF2221, %COI-NEXT UTokyo
JST CREST Grant Number JPMJCR23I4,%CREST
    IBM Quantum, 
 %Institute of AI and Beyond of the University of Tokyo, % NY: letme comment out this
 JST ERATO Grant Number JPMJER2302, 
 and JST ASPIRE Grant Number JPMJAP2316. % ASPIRE

\red{\section*{Author contributions}
K.T. and N.Y. conceived the project and designed the research with inputs from all authors. K.T., Y.M., and N.Y. developed the theoretical framework and performed the analytical derivations. K.T. and N.Y. implemented the numerical simulations. The manuscript was written by K.T, K.S., and N.Y. with contributions from all authors. All authors discussed the results and reviewed the manuscript.}

\red{\section*{Competing Interests}
The authors declare no competing interests.}

\let\oldaddcontentsline\addcontentsline% Store \addcontentsline
\renewcommand{\addcontentsline}[3]{}% Make \addcontentsline a no-op
\bibliography{bib.bib}
\let\addcontentsline\oldaddcontentsline% Restore \addcontentsline
\onecolumngrid
%\appendix

\clearpage
\begin{center}
	\Large
	\textbf{\red{Supplementary Note} for: Symmetric Clifford twirling for cost-optimal quantum error mitigation in early FTQC regime}
\end{center}

\setcounter{section}{0}
\setcounter{equation}{0}
\setcounter{figure}{0}
\setcounter{table}{0}
\setcounter{thm}{0}
\renewcommand{\thesection}{S\arabic{section}}
\renewcommand{\theequation}{S\arabic{equation}}
\renewcommand{\thefigure}{S\arabic{figure}}
\renewcommand{\thetable}{S\arabic{table}}
\renewcommand{\thethm}{S\arabic{thm}}
\renewcommand{\thecor}{S\arabic{cor}}

\addtocontents{toc}{\protect\setcounter{tocdepth}{0}}

\section{Scheme of early fault-tolerant quantum computing}
\label{sec_S_earlyFTQC}
As we have mentioned in the main text, we envision that, in the early FTQC era, the logical error rate is significantly suppressed by quantum error correction so that we may implement deep quantum circuits for quantum information processing tasks such as Hamiltonian simulation and quantum phase estimation.
The question is the condition required to perform symmetric Clifford twirling, i.e., whether we can insert twirling operations without significantly affecting the total error budget.
As discussed in the main text, the discussion relies on how abundant the magic state supply is.
In the following, we discuss the cases of single-thread and multi-thread supply in Secs.~\ref{subsec:single-thread} and~\ref{subsec:multi-thread}, respectively.

\subsection{Single-thread regime}\label{subsec:single-thread}
In this subsection, we argue that, if the supply of magic states is scarce such that non-Clifford operations can be executed only one by one, additional Clifford operations for symmetric Clifford twirling affect the error only in a depth-independent manner,  in particular when the circuit depth is of $\Omega(n)$.

Consider the well-known compilation scheme proposed by Litinski~\cite{litinski2019game}, which follows rules shown in Fig.~\ref{fig_rule} to obtain a circuit as depicted in Fig.~\ref{fig_conjugation}(a).
Namely, all the non-Clifford Pauli rotations are conjugated towards the beginning of the circuit, and the Clifford operations are conjugated so that it can be merged with the measurement.
When the circuit is sufficiently deep, e.g. $\Omega(n)$ depth for $n$-qubit simulation, then most of the multi Pauli rotations and multi Pauli measurements involve \red{$\Omega(n)$} logical qubits with high probability.
In this regime, the amount of total error is no longer affected by adding additional Clifford operations.

An important caveat is that arbitrary symmetric twirling operations reduce to identity under Litinski's compilation. To overcome this issue, we introduce a generalized compilation scheme that allows nontrivial stabilizer state preparation. Concretely, as depicted in Fig.~\ref{fig_conjugation}(b), we first insert twirling operations $\mathcal{D}_l, \mathcal{D}_l^\dag$ for the $l$-th non-Clifford layer $\mathcal{U}_l$. Then, by conjugating $\mathcal{D}_l$ towards the state preparation and $\mathcal{D}_l^\dag$ towards the measurement, we finally obtain the compiled form. 
Since the number and weight of non-Clifford operations and final multi-Pauli measurements is not changed, the additional overhead is only the stabilizer state preparation.

 \begin{figure}[t]
    \begin{center}
        \resizebox{0.75\hsize}{!}
    {\includegraphics{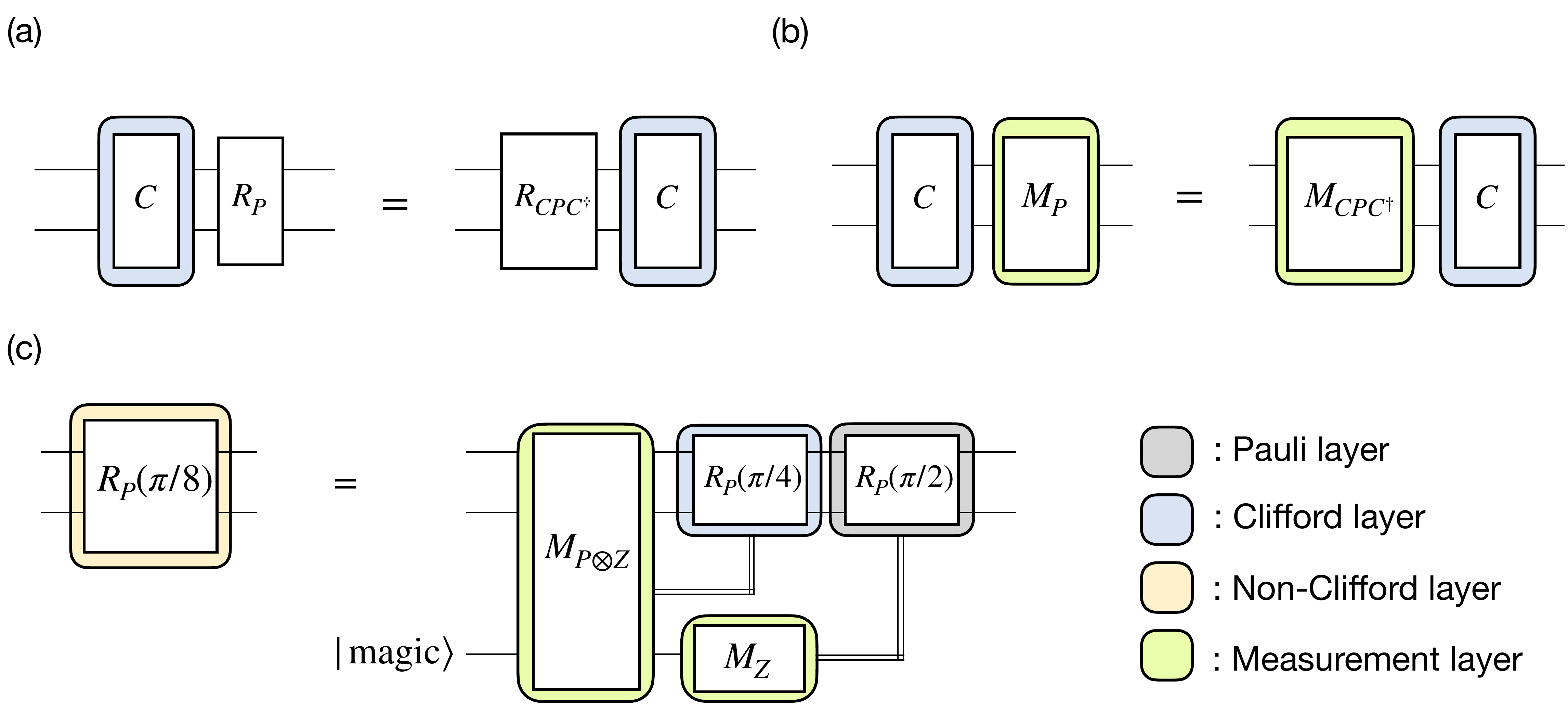}}
        \caption{Gate conjugation rule discussed in Ref.~\cite{litinski2019game}. (a) Commutation rule of Pauli rotation gates. Here, $R_P$ is the rotation gate regarding Pauli operator $P$, and $C$ is a Clifford gate.
        (b) Commutation rule of Pauli measurement. (c) Implementation of $\pi/8$ rotation of Pauli operator $P$ via gate teleportation of $|{\rm magic}\rangle = |0\rangle + e^{i \pi/4}|1\rangle.$ Note that $| {\rm magic}\rangle$ does not rely on $R_P$ as long as the rotation angle is $\pi/8$, which simplifies the magic state distillation protocol.
        }
        \label{fig_rule}
    \end{center}
    \begin{center}
        \resizebox{0.85\hsize}{!}{\includegraphics{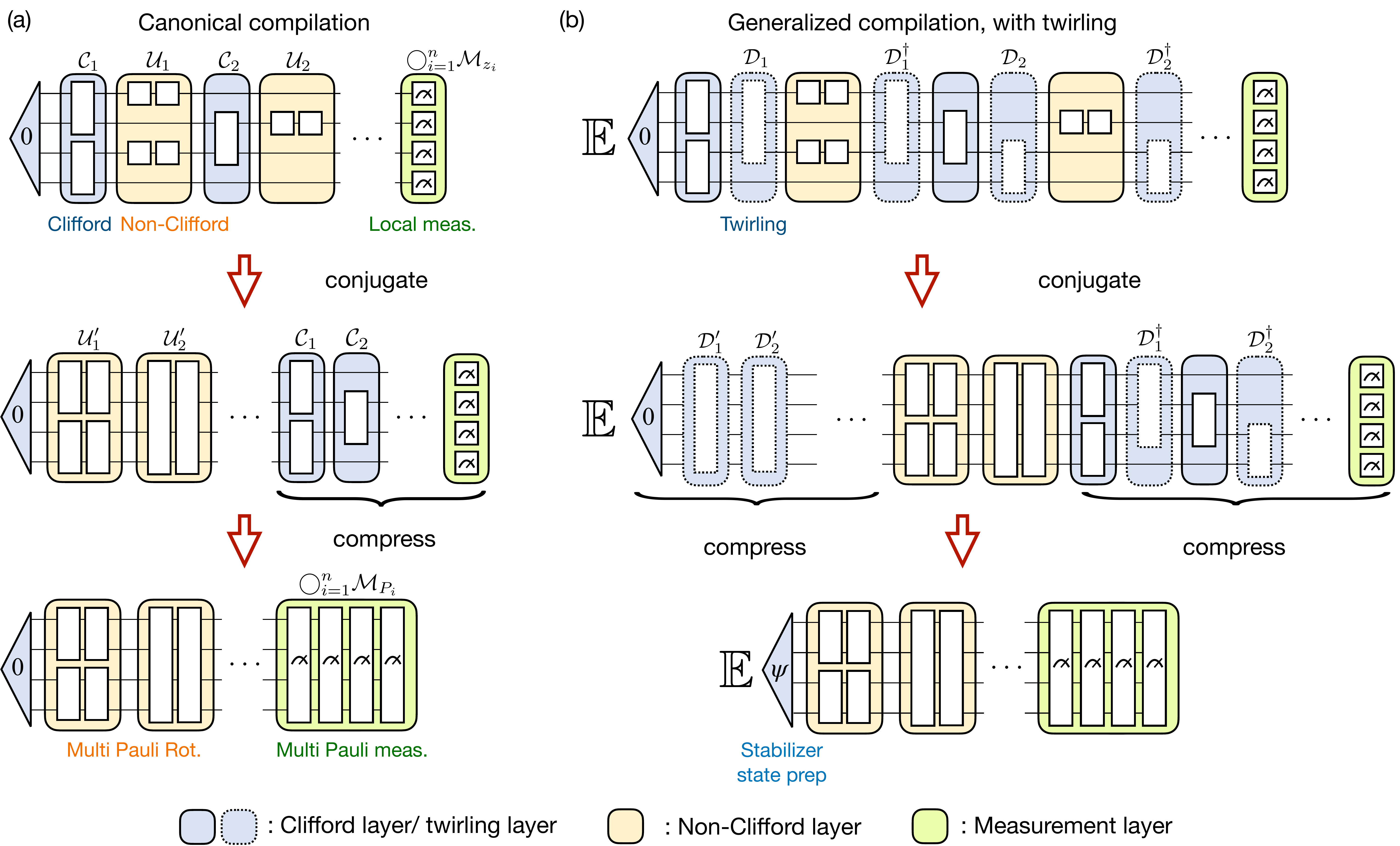}}
        \caption{(a) Canonical compilation by Litinski~\cite{litinski2019game} for untwirled circuit and (b) generalized compilation for twirled quantum circuits. We assume FTQC with lattice surgery and magic state gate teleportation. In the first stage, we conjugate the non-Clifford Pauli rotations through the Clifford layers so that all non-Clifford layers are placed before the Clifford layers. In the second stage, the Clifford layers are absorbed into the stabilizer state preparation (denoted as $\psi$ in (b)) and Pauli measurements. The average $\mathbb{E}$ is taken with respect to the twirling operations.
        }
        \label{fig_conjugation}
    \end{center}
\end{figure}

One naive option for stabilizer state preparation is to perform $n$ steps of Pauli measurement and feedback regarding the stabilizer generators. This introduces an error that scales linearly with $n$.
On the other hand, if $O(n)$ ancilla qubits are available (which is a common assumption in many FTQC architectures~\cite{fowler2018low, gidney2021factor, yoshioka2024hunting}), one may alternatively consider the method proposed by Zheng {\it et al.} that implements Clifford circuits with $O(1)$ depth using gate teleportation~\cite{zheng2018depth}. 
The key idea of achieving the constant depth is to employ the 9-stage decomposition of Clifford circuits, -C-P-C-P-H-P-C-P-C-, where -C-, -P-, and -H- respectively stand for stages consisting only of CNOT gates, phase gates, and Hadamard gates. 
For the sake of stabilizer state preparation, we may fix the initial state to be $|0^n\rangle$, which allows us to consider gate teleportation of only the latter 5 stages, -H-P-C-P-C-. 
Since C, P, and H stages require three, two, and two steps of Pauli measurements with at most two $2n$-qubit CSS states, one $4n$-qubit CSS state, and one $4n$-qubit CSS state, the entire process can be completed by 12 steps of Pauli measurements that consumes three $4n$-qubit stabilizer states and four $2n$-qubit stabilizer states in total. 
This is independent of the number of non-Clifford layers in the original circuit; the circuit volume is not significantly affected, and thus the contribution to the total error can be expected to be small.

Two remarks are in order. 
First, while the preparation of CSS states naively does not seem to be advantageous, it has been shown that such states can be distilled with constant rate~\cite{lai2017fault}, and thus the contribution to the total error can be efficiently reduced.
Second, one may reduce the ancilla count by considering the transversal implementation of some stages or sacrificing the depth by dividing a single step into smaller units.

\subsection{Multi-thread regime}\label{subsec:multi-thread}
Now we discuss the case when the number of magic state factories is sufficient to allow a multi-thread supply of magic. To understand when non-Clifford dominate the total error budget, we decompose the main logical errors into a sum of three terms: decoding error, distillation error, and gate synthesis error.

The decoding error indicates the error due to the limitation of the error correction code and decoding algorithm. When we employ an error correction code of distance $d$, no matter how good the decoding algorithm is, we inevitably have logical error of $p_{\rm dec}=O((p/p_{\rm th})^{\frac{d+1}{2}})$  where $p$ is the physical error rate and $p_{\rm th}$ is the threshold of the code. For instance, in the case of surface code, it is commonly estimated that $p_{\rm dec}=0.1(p/p_{\rm th})^{\frac{d+1}{2}}$ with $p_{\rm th}=0.01$~\cite{litinski2019game, lee2021even, yoshioka2024hunting} and therefore we have $p_{\rm dec} = 10^{-8}$ for $p=10^{-3}$ and $d=13$, for instance.
When we perform a multi-Pauli measurement that involves $k$ logical qubits, the logical error is multiplied by a factor of $k$. Therefore, if the circuit volume is given by $V$, the total decoding error is given as
\begin{eqnarray}
    N_{\rm dec} = \red{d}p_{\rm dec} V, \label{eq:n_dec}
\end{eqnarray}
\red{where $d$ is multiplied to reflect the fact that stabilizer measurement is repeated for $d$ rounds in many quantum error correcting codes including surface code.}

The distillation error is rooted in the insufficiency of the magic state distillation.
For instance, in the context of preparing the magic state for T gates, it is common to utilize the 15-to-1 protocol that is essentially an error detection using Reed-Muller code that allows transversal implementation of the T gate~\cite{bravyi2005universal}.
Reflecting the fact that the code is a $\llbracket 15,1,3\rrbracket$ code, the protocol suppresses the error up to the cubic order and outputs a single magic state of error rate $35p^3$, using 15 noisy magic states of error rate $p.$ 
In general, by choosing $\llbracket n',k',d' \rrbracket$ codes that allows transversal implementation of T gates, one can suppress the distillation error up to $d'$th order. 
The total of the distillation error can be written as
\begin{eqnarray}
    N_{\rm dis} = p_{\rm dis}n_T,\label{eq:n_dis}
\end{eqnarray}
where $p_{\rm dis}$ is the error rate of the distilled magic state and $n_T$ is the number of $\pi/8$ rotations (or T count) in the circuit.
If multiple types of magic states are employed, we can simply sum over the error from each of them.

Finally, the synthesis error indicates the coherent error that arises when we decompose a given arbitrary gate into a sequence of fault-tolerant gates that only constitute a discrete set. For instance, when one employs the Clifford+T gate set, a Pauli rotation gate is typically decomposed into a sequence that consists of H, S, and T gates.
The gate complexity to approximate the target unitary up to the accuracy of $\epsilon$ is $O({\rm polylog}(1/\epsilon)),$ which is considered to be non-negligible when the number of T gates are limited as in early FTQC era.
However, the synthesis error is distinct from other two errors in a sense that it can be characterized completely without quantum measurement.
This has led to several proposals that suppress the synthesis error up to quadratic or cubic order without increasing circuit depth at all~\cite{hastings2016turning, campbell_shorter_2017, yoshioka2024error}.
For the sake of the discussion below, let us define the total contribution from synthesis error as $N_{\rm syn}.$ In the case when the quantum circuit consists of Clifford gates and Pauli rotations, we can rewrite it as 
\begin{eqnarray}
N_{\rm syn} = p_{\rm rot}N_{\rm rot},    
\end{eqnarray}
 where $p_{\rm rot}$ is the synthesis error per gate and $N_{\rm rot}$ is the number of rotation gates.

With all the contributions taken into account, the expected number of errors in the quantum circuit per single run can be written as
\begin{eqnarray}
    N_{\rm err} = N_{\rm dec} + N_{\rm dis} + N_{\rm syn}.
\end{eqnarray}
When we run the quantum circuit, we design the entire protocol so that  $N_{\rm err} = O(1)$, while we can arbitrarily choose how each error contributes to $N_{\rm err}.$
In this work, we consider a situation where the number of physical qubits is limited to tens of thousands to hundreds of thousands.
In such a situation, we may allow increasing the code distance $d$ for each logical qubit, while we do not want to blow up the number of logical qubits. Since the number of data and ancilla logical qubits is typically determined from the problem itself, we shall aim to reduce the number of logical qubits for magic state factories.
For instance, we may use the zero-level distillation that yields an error of $100 p^2$ by using only a single logical qubit~\cite{itogawa2024even}, or may choose one round of magic state distillation that yields $p_{\rm dis}=O(p^2).$ In the case of $p=10^{-3},$ we typically have $p_{\rm dis}=10^{-4}$ or $10^{-5}.$ This is orders of magnitudes higher than the decoding error assuming $d=11$ or $d=13$, and then the total error is dominated by the non-Clifford gates.

Regarding the T states, a recent proposal has shown that the circuit volume for magic state preparation has been significantly improved to be comparable with CNOT gates~\cite{gidney2024magic}.
However, if the non-Clifford operations consist of heavily-magic-consuming gates such as the Pauli rotation, the error from non-Clifford operations dominates the total error budget.
In such a case, low-cost substitutes of full symmetric twirling, such as the $k$-sparse twirling with $k=O(1)$, introduce a negligible amount of error.
It is noteworthy that such a setup includes a wealth of quantum algorithms such as the Trotterized Hamiltonian simulation (both time-independent~\cite{childs2018toward} and time-dependent~\cite{mizuta2024explicit}), adiabatic state preparation~\cite{aspuru2005simulated}, and quantum phase estimation~\cite{kivlichan2020improved, campbell_shorter_2017}.

\section{Comparison of the sampling overhead of QEM methods}
\label{sec_S_QEMcost}
\begin{figure}[t]
    \begin{center}
        \includegraphics[width=0.5\linewidth]{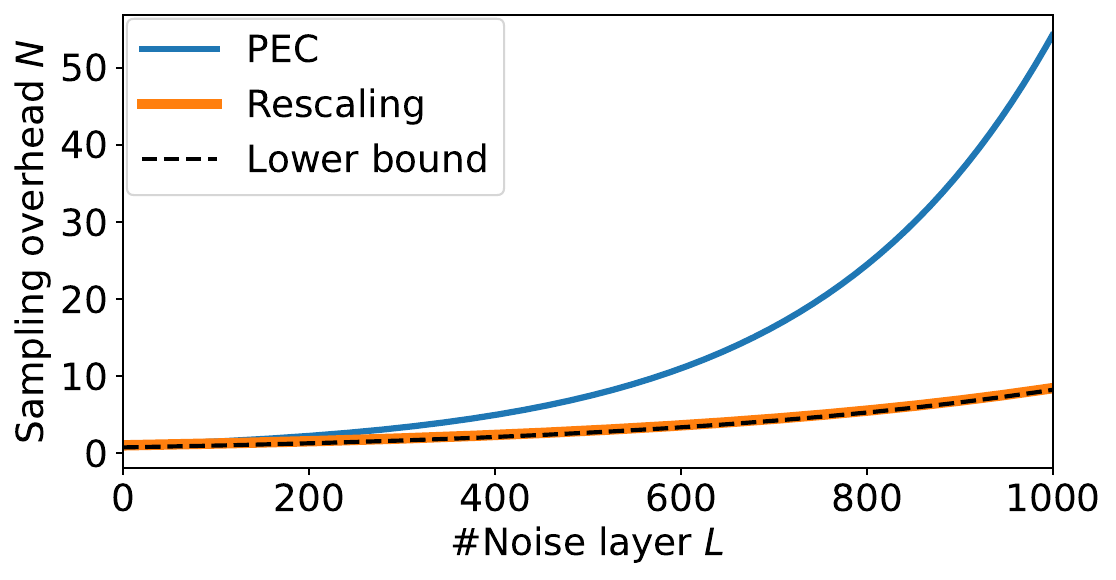}
        \caption{Comparison of the sampling overhead $N$ for probabilistic error cancellation~\cite{piveteau2021error, lostaglio2021error, suzuki2022quantum} (blue line) and the rescaling method (orange line) when $p_{\mathrm{err}} = 10^{-3}$ and $n\gg 1$. The sampling overhead of the rescaling method is not only quadratically improved from probabilistic error cancellation, but also saturates the lower bound proposed in Ref. \cite{tsubouchi2023universal} (dotted line).
        }
        \label{fig_cost}
    \end{center}
\end{figure}

In this section, we discuss the sampling overhead of some quantum error mitigation (QEM) methods and compare them with the theoretical lower bound proposed in Ref.~\cite{tsubouchi2023universal}.

Let us focus on the scenario where we aim to estimate the expectation value of an observable $O$ for the $n$-qubit 
noiseless quantum state
\begin{equation}
    \rho_{\mathrm{ideal}} := \mathcal{C}_{L+1}\circ\mathcal{U}\circ\mathcal{C}_{L}\circ \cdots \circ\mathcal{U}\circ\mathcal{C}_{1}(\rho_0),
\end{equation}
but we only have access to the noisy quantum state
\begin{equation}
    \rho_{\mathrm{noisy}} := \mathcal{C}_{L+1}\circ\mathcal{N}\circ\mathcal{U}\circ\mathcal{C}_{L}\circ \cdots \circ\mathcal{N}\circ\mathcal{U}\circ\mathcal{C}_{1}(\rho_0),
\end{equation}
where $\mathcal{C}_l(\cdot) := C_l \cdot C_l^\dagger$ is the $l$-th Clifford layer, $\mathcal{U}(\cdot) := U \cdot U^\dagger$ is the non-Clifford layer, $\mathcal{N}$ is the noise layer, $\rho_0$ is the initial state, and $L$ represents the number of the noisy non-Clifford layers. Note that this expression can be obtained by assuming that (i) noise in Clifford operation is negligible compared to non-Clifford operations, (ii) all non-Clifford operations are identical up to Clifford conjugation.
Here, for simplicity, we assume that $O$ is a traceless observable satisfying $-I\leq O \leq I$ and $\mathcal{N}$ is a Pauli noise given by $\mathcal{N} := (1-p_{\mathrm{err}})\mathcal{I} + \sum_i p_i\mathcal{E}_{P_i}$, where Pauli error $\mathcal{E}_{P_i}(\cdot) := P_i\cdot P_i$ occurs with probability $p_i$ with $p_{\mathrm{err}} := \sum_i p_i$ being the error probability.

By utilizing QEM techniques, we can construct an estimator for $\mathrm{tr}[\rho_{\mathrm{ideal}}O]$ solely from multiple executions of the noisy logical circuit (with some modifications).
Recently, a lower bound on the sampling overhead $N$, which is a multiplicative factor in the number of circuit executions needed to achieve the same estimation accuracy as in the noiseless case, has been established as
\begin{equation}
    N \geq \qty(\frac{4^n\nu(\mathcal{N}^{-1})-1}{4^n-1})^{L} - \frac{2^n-2}{4^n-1},
\end{equation}
where $\nu(\mathcal{N}^{-1}) = \mathrm{tr}[I\otimes\mathcal{N}^{-1}(\ketbra{\Gamma})^2]/4^n$ represents the purity of the Choi matrix of $\mathcal{N}^{-1}$ with $\ket{\Gamma} = \sum_i \ket{i}\ket{i}$ (Theorem~S1 of Ref.~\cite{tsubouchi2023universal}).
When we assume that $\mathcal{N}$ is a Pauli noise, we can calculate $\nu(\mathcal{N}^{-1})$ as $\nu(\mathcal{N}^{-1}) = 1 + 2p_{\mathrm{err}} + O(p_i^2)$.
Consequently, we obtain 
\begin{equation}
    \label{eq_QEMbound}
    N \gtrsim \qty(1+ \frac{4^n}{4^n-1}2p_{\mathrm{err}})^{L} - \frac{2^n-2}{4^n-1}
\end{equation}
as the lower bound on the sampling overhead necessary to mitigate Pauli errors, considering up to the first degree of the error probability.

This lower bound can be saturated when we can convert noise layer $\mathcal{N}$ into the global white noise defined as 
\begin{equation}
    \mathcal{N}_{\mathrm{wn},p_{\mathrm{err}}} := (1-p_{\mathrm{err}})\mathcal{I} + p_{\mathrm{err}}\mathbb{E}_{P\in\mathcal{P}_n-\qty{I}^{\otimes n}}[\mathcal{E}_{P}],
\end{equation}
where each $n$-qubit Pauli error $P\in\mathcal{P}_n-\qty{I}^{\otimes n}$ occurs with equal probability of $(4^n-1)^{-1}p_{\mathrm{err}}$.
Global white noise can also be expressed as
\begin{equation}
    \mathcal{N}_{\mathrm{wn},p_{\mathrm{err}}}(\cdot) = \qty(1-\frac{4^n}{4^n-1}p_{\mathrm{err}})\cdot + \frac{4^n}{4^n-1}p_{\mathrm{err}}I,
\end{equation}
so we can mitigate errors by simply rescaling the noisy expectation value as
\begin{equation}
    \qty(1-\frac{4^n}{4^n-1}p_{\mathrm{err}})^{-L}\mathrm{tr}[\rho_{\mathrm{noisy}}O] = \mathrm{tr}[\rho_{\mathrm{ideal}}O]
\end{equation}
with a sampling overhead of $N = \qty(1-\frac{4^n}{4^n-1}p_{\mathrm{err}})^{-2L}$.
This rescaling method represents a quadratic improvement over probabilistic error cancellation, whose sampling overhead is $N \sim \qty(1+ 2p_{\mathrm{err}})^{2L}$~\cite{piveteau2021error, suzuki2022quantum}.
Additionally, this sampling overhead saturates the lower bound given by Eq. (\ref{eq_QEMbound}) when $p_{\mathrm{err}}$ is small and $n$ is large. 
Thus, we can conclude that the rescaling method is a cost-optimal QEM method (see Fig. \ref{fig_cost}).

While the rescaling method effectively mitigate errors with minimal sampling overhead, this overhead still grows exponentially with the number of noisy layer $L$.
Therefore, a practical scenario arises when the total error rate $p_{\mathrm{tot}} := p_{\mathrm{err}}L$ can be regarded as a constant value.
In this situation,  if $p_{\mathrm{err}}$ is small and $n$ is large, the sampling overhead can be approximated as a constant value represented as $N\sim e^{2p_{\mathrm{tot}}}$.

We note that the lower bound described in Eq. (\ref{eq_QEMbound}) applies solely to QEM methods that do not utilize additional ancilla qubits or implement mid-circuit measurements.
Indeed, by introducing logical ancilla qubits, we can reduce the sampling overhead below the lower bound by performing error detection.
However, the use of logical ancilla qubits may impose an extra burden in terms of the space overhead, which may not be desirable when the number of qubits is scarce.
Therefore, we focus on the QEM methods that do not use ancilla qubits and regard the rescaling method as a cost-optimal QEM method.

\section{Generalization and proof of Theorem 1 in the main text}
In this section, we generalize Theorem~1 in the main text to arbitrary non-Clifford unitary $U$, so that we can apply symmetric Clifford twirling to any non-Clifford gates such as the Toffoli gate or multi-qubit Pauli rotation gates.
As in the main text, we define the Pauli subgroup $\mathcal{Q}_{U}$ as
\begin{equation}
    \label{eq_Q_U_supp}
    \mathcal{Q}_{U} :=\ev*{\qty{P\in\mathcal{P}_n\;|\;\mathrm{tr}[PU]\neq0}},
\end{equation}
where $\ev*{\cdot}$ represents the group generated by the elements within the brackets.
In Ref. \cite{mitsuhashi2023clifford}, it was shown that arbitrary Pauli subgroup $\mathcal{Q}_{U}$ can be decomposed into three parts as
\begin{equation}
    \label{eq_Pauli_subgroup_2}
    \mathcal{Q}_{U}  = W^\dagger(\qty{\mathrm{I},\mathrm{X},\mathrm{Y},\mathrm{Z}}^{\otimes n_1} \otimes \qty{\mathrm{I},\mathrm{Z}}^{\otimes n_2} \otimes \qty{\mathrm{I}}^{\otimes n_3})W
\end{equation}
up to phase, using a $\mathcal{Q}_{U}$-dependent Clifford unitary $W \in \mathcal{G}_n$ and $n_1 + n_2 + n_3 = n$.
For example,  $W = I$ and $n_1=0, n_2 = 1, n_3 = n-1$ hold for the T gate and the $R_z(\theta)$ gate, $n_1=0, n_2 = 3, n_3 = n-3$ and $W$ is the Hadamard gate acting on the third qubit for the Toffoli gate, and $W$ is the CNOT gate with $n_1=0,n_2=1,n_3=n-1$ when $U$ is the Pauli rotation gate of $Z^{\otimes 2}$.
By using the decomposition in Eq. (\ref{eq_Pauli_subgroup_2}), we obtain the following theorem.

\begin{thm}[general case]
    \label{thm_2_supply}
    Let $P \in \mathcal{P}_n$ be an $n$-qubit Pauli operator, $\mathcal{E}_P(\cdot) = P\cdot P$ be the corresponding Pauli channel, and $\mathcal{Q}_{U}$ be a Pauli subgroup defined as in Eq.~(\ref{eq_Q_U_supp}). 
    Let us decompose $\mathcal{Q}_{U}$ as in Eq. (\ref{eq_Pauli_subgroup_2}) and represent $P$ corresponding to $\mathcal{E}_P$ as $P = W^\dagger(P_1\otimes P_2 \otimes P_3)W$ up to phase, where $P_1\in \mathcal{P}_{n_1}$, $P_2\in \mathcal{P}_{n_2}$, $P_3\in \mathcal{P}_{n_3}$, and $W$ is a Clifford operator defined in Eq.~\eqref{eq_Pauli_subgroup_2}.
   Then, the Pauli channel scrambled through symmetric Clifford twirling $\mathscr{T}_{\mathcal{Q}_{U}}(\mathcal{E}_P) = \mathbb{E}_{D\in\mathcal{G}_{n, \mathcal{Q}_U}}[\mathcal{D}^\dagger\circ\mathcal{E}_P\circ\mathcal{D}]$ can be represented as:
    \begin{itemize}
        \item[(i)] when $P_2\notin \qty{\mathrm{I}, \mathrm{Z}}^{\otimes n_2}$,
        \begin{equation}
            \mathscr{T}_{\mathcal{Q}_{U}}(\mathcal{E}_P) 
            = \underset{\substack{{Q_2\in\qty{\mathrm{I}, \mathrm{Z}}^{\otimes n_2}} \\ {Q_3\in\mathcal{P}_{n_3}}}}{\mathbb{E}}\qty[\mathcal{E}_{W^\dagger(P_1\otimes P_2Q_2 \otimes Q_3)W}],
        \end{equation}
        \item[(ii)] when $P_2\in \qty{\mathrm{I}, \mathrm{Z}}^{\otimes n_2}$ and $P_3 \neq \mathrm{I}^{\otimes n_3}$,
        \begin{equation}
            \mathscr{T}_{\mathcal{Q}_{U}}(\mathcal{E}_P) 
            =\underset{\substack{Q_2\in\qty{\mathrm{I}, \mathrm{Z}}^{\otimes n_2} \\ Q_3\in\mathcal{P}_{n_3}\setminus\qty{\mathrm{I}}^{\otimes n_3}}}{\mathbb{E}}\qty[\mathcal{E}_{W^\dagger(P_1\otimes Q_2 \otimes Q_3)W}],
        \end{equation}
        \item[(iii)] when $P_2\in \qty{\mathrm{I}, \mathrm{Z}}^{\otimes n_2}$ and $P_3 = \mathrm{I}^{\otimes n_3}$,
        \begin{equation}
            \mathscr{T}_{\mathcal{Q}_{U}}(\mathcal{E}_P) =\mathcal{E}_P.
        \end{equation}
    \end{itemize}
\end{thm}

By using Theorem~\ref{thm_2_supply}, we can discuss how the noise affecting general non-Clifford unitary such as the Toffoli gate or multi-qubit Pauli rotation gates 
can be scrambled through symmetric Clifford twirling.
We note that Theorem~1 in the main text directly follows from Theorem~\ref{thm_2_supply} by setting $W = I$ and $n_1=0, n_2 = 1, n_3 = n-1$.

\begin{proof}
In Ref. \cite{mitsuhashi2023clifford}, it was shown that all elements in $\mathcal{Q}_U$-symmetric Clifford group
\begin{equation}
    \mathcal{G}_{n,\mathcal{Q}_{U}} := \qty{C\in\mathcal{G}_{n} \;|\; \forall P \in \mathcal{Q}_{U}, \;[C,P]=0}
\end{equation}
can be uniquely represented as
\begin{equation}
    \label{eq_Clifford_decomposition}
    D = W^\dagger (\mathrm{I}^{\otimes n_1} \otimes D_1) \times  (\mathrm{I}^{\otimes n_1} \otimes D_2 \otimes \mathrm{I}^{\otimes n_3})\times (\mathrm{I}^{\otimes n_1}\otimes \mathrm{I}^{\otimes n_2}\otimes D_3) W
\end{equation}
up to the phase.
Here, $D_1$ is an $n_2+n_3$-qubit unitary represented as 
\begin{equation}
    \label{eq_D_1}
    D_1 = \Lambda_1(R_1) \cdots \Lambda_{n_2}(R_{n_2}),
\end{equation}
where $R_i\in\mathcal{P}_{n_3}$ is a Pauli operator acting on the last $n_3$ qubits and $\Lambda_i(R_i)$ is the controlled-$R_i$ gate whose controlled qubit is the $i$-th qubit.
$D_2$ is a element of  $\qty{\mathrm{I},\mathrm{Z}}^{\otimes n_2}$-symmetric Clifford group
\begin{equation}
    D_2\in\mathcal{G}_{n_2,\qty{\mathrm{I},\mathrm{Z}}^{\otimes n_2}} := \qty{C\in\mathcal{G}_{n} \;|\; \forall R \in \qty{\mathrm{I},\mathrm{Z}}^{\otimes n_2}, \;[C,R]=0}
\end{equation}
and it can be uniquely represented as 
\begin{equation}
    \label{eq_D_2}
    D_2 =\prod_{\substack{i,j\in\qty{1,\cdots,n_2}\\i<j}}\mathrm{CZ}_{ij}^{\nu_{ij}} \prod_{i\in\qty{1,\cdots,n_2}}\mathrm{S}_{i}^{\mu_{i}},
\end{equation}
where $\mathrm{CZ}_{ij}$ is the CZ gate acting on the $i$-th and $j$-th qubits, $\mathrm{S}_i$ is the S gate acting on $i$-th qubit, $\nu_{ij}\in\qty{0,1}$, and $\mu_{i}\in\qty{0,1,2,3}$.
$D_3$ is an element of the Clifford group $\mathcal{G}_{n_3}$:
\begin{equation}    
    \label{eq_D_3}
    D_3 \in \mathcal{G}_{n_3}.
\end{equation}

Let us define the superchannels
\begin{eqnarray}
    \label{eq_twirling_decomposition}
    \mathscr{T}_1(\mathcal{N}) &:=& \mathbb{E}_{D_1} \qty[\mathcal{E}_{\mathrm{I}^{\otimes n_1} \otimes D_1}^\dagger\circ\mathcal{N}\circ\mathcal{E}_{\mathrm{I}^{\otimes n_1} \otimes D_1}],\\
    \mathscr{T}_2(\mathcal{N}) &:=& \mathbb{E}_{D_2} \qty[\mathcal{E}_{\mathrm{I}^{\otimes n_1} \otimes D_2\otimes\mathrm{I}^{\otimes n_3}}^\dagger\circ\mathcal{N}\circ\mathcal{E}_{\mathrm{I}^{\otimes n_1} \otimes D_2\otimes\mathrm{I}^{\otimes n_3}}],\\
    \mathscr{T}_3(\mathcal{N}) &:=& \mathbb{E}_{D_3} \qty[\mathcal{E}_{\mathrm{I}^{\otimes n_1} \otimes \mathrm{I}^{\otimes n_2} \otimes D_3}^\dagger\circ\mathcal{N}\circ\mathcal{E}_{\mathrm{I}^{\otimes n_1} \otimes \mathrm{I}^{\otimes n_2} \otimes D_3}],
\end{eqnarray}
where $\mathcal{E}_U(\cdot) := U\cdot U^\dagger$ and $\mathbb{E}_{D_i}$ denotes the average over all operators given by Eqs. (\ref{eq_D_1}), (\ref{eq_D_2}), and (\ref{eq_D_3}).
From Eq. (\ref{eq_Clifford_decomposition}), we can decompose the symmetric Clifford twirling superchannel $\mathscr{T}_{\mathcal{Q}_U}$ as
\begin{equation}
    \mathscr{T}_{\mathcal{Q}_U}(\mathcal{E}_P)
    =\mathcal{E}_{W^\dag}\circ(\mathscr{T}_{3}\ast\mathscr{T}_{2}\ast\mathscr{T}_{1}(\mathcal{E}_{W PW^\dag}))\circ\mathcal{E}_{W},
    \label{eq_superchannel_decomposition}
\end{equation}
where we define the composition of superchannels as $\mathscr{T}_2\ast\mathscr{T}_1(\cdot) := \mathscr{T}_2(\mathscr{T}_1(\cdot))$.
Equation~\eqref{eq_superchannel_decomposition} implies that it is sufficient to consider $\mathscr{T}_3\ast\mathscr{T}_2\ast\mathscr{T}_1(\mathcal{E}_{W PW^\dag})$ in order to know $\mathscr{T}_{\mathcal{Q}_U}(\mathcal{E}_P)$. 
Since $W$ is a Clifford unitary, $WPW^\dag $ is a Pauli operator and can be decomposed as $W PW^\dag=P_1\otimes P_2\otimes P_3$ up to phase with $P_1\in\mathcal{P}_{n_1}$, $P_2\in\mathcal{P}_{n_2}$, and $P_3\in\mathcal{P}_{n_3}$ in the same way as in Eq.~\eqref{eq_Pauli_subgroup_2}. 
Below, we see how the Pauli channel $\mathcal{E}_{W PW^\dag}$ is mapped by the superchannels in Eq. (\ref{eq_superchannel_decomposition}), especially in the following four cases.

First, we consider the case where $P_2\in \qty{\mathrm{I},\mathrm{Z}}^{\otimes n_2}$ and $P_3 = \mathrm{I}^{\otimes n_3}$.
Since $\mathrm{I}^{\otimes n_1} \otimes D_1$, $\mathrm{I}^{\otimes n_1} \otimes D_2 \otimes \mathrm{I}^{\otimes n_3}$, and $\mathrm{I}^{\otimes n_1}\otimes \mathrm{I}^{\otimes n_2}\otimes D_3$ all commutes with $WPW^\dag = P_1\otimes P_2 \otimes \mathrm{I}^{\otimes n_3}$ for $P_2\in \qty{\mathrm{I},\mathrm{Z}}^{\otimes n_2}$, we can easily see that
\begin{equation}
    \mathscr{T}_{3}\ast\mathscr{T}_{2}\ast\mathscr{T}_{1}(\mathcal{E}_{ P_1\otimes P_2 \otimes \mathrm{I}^{\otimes n_3}}) = \mathcal{E}_{ P_1\otimes P_2 \otimes \mathrm{I}^{\otimes n_3}}. \label{eq_twirling_case1}
\end{equation}

Next, we consider the case where $P_2\in \qty{\mathrm{I},\mathrm{Z}}^{\otimes n_2}$ and $P_3 \neq \mathrm{I}^{\otimes n_3}$. 
Since $P_1\otimes P_2 \otimes \mathrm{I}^{\otimes n_3}$ commutes with all $D_i$, we only need to consider the case with $WPW^\dag = \mathrm{I}^{\otimes n_1}\otimes \mathrm{I}^{\otimes n_2}\otimes P_3$.
Since $\Lambda_i(R_i)^\dagger P_3 \Lambda_i(R_i) = P_3$ when $[R_i, P_3] = 0$, $\Lambda_i(R_i)^\dagger P_3 \Lambda_i(R_i) = \mathrm{Z}_iP_3$ when $\{R_i, P_3\} = 0$, and Pauli operator $R_i$ which commutes or anticommutes with $P_3$ appears with equal frequency under the average of $D_1$, we obtain
\begin{equation}
    \mathscr{T}_{1}(\mathcal{E}_{\mathrm{I}^{\otimes n_1}\otimes \mathrm{I}^{\otimes n_2} \otimes P_3}) = \mathbb{E}_{Q_{2} \in \qty{\mathrm{I},\mathrm{Z}}^{\otimes n_2}}[\mathcal{E}_{\mathrm{I}^{\otimes n_1} \otimes Q_2 \otimes P_3}].
\end{equation}
In addition, since $D_2$ represented as in Eq. (\ref{eq_D_2}) commutes with  Pauli-Z operator, we obtain
\begin{equation}
    \mathscr{T}_{2}(\mathcal{E}_{\mathrm{I}^{\otimes n_1} \otimes Q_2 \otimes P_3}) = \mathcal{E}_{\mathrm{I}^{\otimes n_1} \otimes Q_2 \otimes P_3}
\end{equation}
for $Q_2\in\qty{\mathrm{I},\mathrm{Z}}^{\otimes n_2}$.
Furthermore, since $\mathscr{T}_{3}$ is just a Clifford twirling superchannel acting on the last $n_3$ qubits, we obtain
\begin{equation}
    \mathscr{T}_{3}(\mathcal{E}_{\mathrm{I}^{\otimes n_1} \otimes Q_2 \otimes P_3}) = \mathbb{E}_{Q_3\in\mathcal{P}_{n_3} \setminus \qty{\mathrm{I}}^{\otimes n_3}}[\mathcal{E}_{\mathrm{I}^{\otimes n_1} \otimes Q_2 \otimes Q_3}].
\end{equation}
Therefore, we have
\begin{equation}
    \mathscr{T}_3\ast\mathscr{T}_2\ast\mathscr{T}_1(\mathcal{E}_{\mathrm{I}^{\otimes n_1} \otimes \mathrm{I}^{\otimes n_2} \otimes P_3}) = \mathbb{E}_{Q_{2} \in \qty{\mathrm{I},\mathrm{Z}}^{\otimes n_2}}\mathbb{E}_{Q_3\in\mathcal{P}_{n_3} \setminus \qty{\mathrm{I}}^{\otimes n_3}}[\mathcal{E}_{\mathrm{I}^{\otimes n_1} \otimes Q_2 \otimes Q_3}]. \label{eq_twirling_case2}
\end{equation}

Then, we consider the case where $P_{2}  \notin \qty{\mathrm{I},\mathrm{Z}}^{\otimes n_2}$ and $P_3 = \mathrm{I}^{\otimes n_3}$.
Since $P_1\otimes P_{2,z} \otimes \mathrm{I}^{\otimes n_3}$ commutes with all $D_i$ for $P_{2,z}  \in \qty{\mathrm{I},\mathrm{Z}}^{\otimes n_2}$, we only need to consider the case with $WPW^\dag = \mathrm{I}^{\otimes n_1}\otimes P_{2,x}\otimes \mathrm{I}^{\otimes n_3}$ for $P_{2,x}  \in \qty{\mathrm{I},\mathrm{X}}^{\otimes n_2} \setminus \qty{\mathrm{I}}^{\otimes n_2}$.
Since $\Lambda_i(R_i)^\dagger \mathrm{X}_i \Lambda_i(R_i) = \mathrm{X}_iR_i$, where $\mathrm{X}_i$ being the Pauli-X operator acting on the $i$-th qubit, and all $R_i\in\mathcal{P}_{n_3}$ appears with equal frequency under the average of $D_1$, we obtain
\begin{equation}
    \mathscr{T}_{1}(\mathcal{E}_{\mathrm{I}^{\otimes n_1} \otimes P_{2,x} \otimes \mathrm{I}^{\otimes n_3}}) = \mathbb{E}_{Q_3\in\mathcal{P}_{n_3}}[\mathcal{E}_{\mathrm{I}^{\otimes n_1} \otimes P_{2,x} \otimes Q_3}].
\end{equation}
In addition, since $\mathrm{CZ}_{12}\mathrm{X}_1\mathrm{CZ}_{12}^\dagger = \mathrm{X}_1\mathrm{Z}_2$ and $\mathrm{S}\mathrm{X}\mathrm{S}^\dagger = i\mathrm{XZ}$, Pauli-X operator belonging to the middle $n_2$ qubit generates Pauli-Z operator with probability 1/2 for each of the $n_2$ qubit under the conjugation of random $D_2$, and thus 
\begin{equation}
    \mathscr{T}_{2}(\mathcal{E}_{\mathrm{I}^{\otimes n_1} \otimes P_{2,x} \otimes Q_3}) = \mathbb{E}_{Q_{2} \in \qty{\mathrm{I},\mathrm{Z}}^{\otimes n_2}}[\mathcal{E}_{\mathrm{I}^{\otimes n_1} \otimes P_{2,x}Q_2 \otimes Q_3}].
\end{equation}
Furthermore, since $\mathscr{T}_{3}$ is just a Clifford twirling superchannel acting  on the last $n_3$ qubits, we obtain
\begin{equation}
    \mathscr{T}_{3}(\mathcal{E}_{\mathrm{I}^{\otimes n_1} \otimes P_{2,x}Q_2 \otimes Q_3}) = \mathbb{E}_{Q_3'\in\mathcal{P}_{n_3}}[\mathcal{E}_{\mathrm{I}^{\otimes n_1} \otimes P_{2,x}Q_2 \otimes Q_3'}].
\end{equation}
Therefore, we have
\begin{equation}
     \mathscr{T}_3\ast\mathscr{T}_2 \ast\mathscr{T}_1(\mathcal{E}_{\mathrm{I}^{\otimes n_1} \otimes P_{2x} \otimes \mathrm{I}^{\otimes n_3}}) = \mathbb{E}_{Q_{2} \in \qty{\mathrm{I},\mathrm{Z}}^{\otimes n_2}}\mathbb{E}_{Q_3\in\mathcal{P}_{n_3}}[\mathcal{E}_{\mathrm{I}^{\otimes n_1} \otimes P_{2x}Q_2 \otimes Q_3}]. \label{eq_twirling_case3}
\end{equation}

Finally, we consider the case where $P_{2}  \notin \qty{\mathrm{I},\mathrm{Z}}^{\otimes n_2}$ and $P_3 \neq \mathrm{I}^{\otimes n_3}$.
Since $P_1\otimes P_{2,z} \otimes \mathrm{I}^{\otimes n_3}$ commutes with all $D_i$ for $P_{2,z}  \in \qty{\mathrm{I},\mathrm{Z}}^{\otimes n_2}$, we only need to consider the case with $WPW^\dag = \mathrm{I}^{\otimes n_1}\otimes P_{2,x}\otimes P_3$ for $P_{2,x} \in \qty{\mathrm{I},\mathrm{X}}^{\otimes n_2} \setminus \qty{\mathrm{I}}^{\otimes n_2}$.
Under the conjugation of $D_1$, we can see that $P = \mathrm{I}^{\otimes n_1}\otimes P_{n_2,x}\otimes P_3$ is mapped in the form of 
\begin{equation}
    (\mathrm{I}^{\otimes n_1} \otimes D_1) (\mathrm{I}^{\otimes n_1}\otimes P_{2,x}\otimes P_3) (\mathrm{I}^{\otimes n_1} \otimes D_1)^{\dagger} = \mathrm{I}^{\otimes n_1} \otimes P_{2,x}Q_2 \otimes Q_3,
\end{equation}
where $Q_{2} \in \qty{\mathrm{I},\mathrm{Z}}^{\otimes n_2}$ and $Q_{3} \in \mathcal{P}_{n_3}$.
Here, all $Q_{3} \in \mathcal{P}_{n_3}$ appears with equal probability when we randomly choose $D_1$.
When we think of applying $\mathscr{T}_3\ast\mathscr{T}_2$ to this Pauli channel, we obtain
\begin{equation}
    \mathscr{T}_3\ast\mathscr{T}_2(\mathcal{E}_{ \mathrm{I}^{\otimes n_1} \otimes P_{2,x}Q_2 \otimes Q_3}) = 
    \left\{
        \begin{array}{ll}
        \mathbb{E}_{Q_{2}' \in \qty{\mathrm{I},\mathrm{Z}}^{\otimes n_2}}\mathbb{E}_{Q_3'\in\mathcal{P}_{n_3} \setminus \qty{\mathrm{I}}^{\otimes n_3}}[\mathcal{E}_{\mathrm{I}^{\otimes n_1} \otimes P_{2,x}Q_2' \otimes Q_3'}]. & (Q_3 \neq \mathrm{I}^{\otimes n_3})\\
        \mathbb{E}_{Q_{2}' \in \qty{\mathrm{I},\mathrm{Z}}^{\otimes n_2}}[\mathcal{E}_{\mathrm{I}^{\otimes n_1} \otimes P_{2,x}Q_2' \otimes \mathrm{I}^{\otimes n_3}}]. & (Q_3 = \mathrm{I}^{\otimes n_3})
        \end{array}
    \right.
\end{equation}
from the discussion of the other case.
By recalling that every $Q_{3} \in \mathcal{P}_{n_3}$ appears with equal probability when we randomly choose $D_1$, we obtain
\begin{equation}
     \mathscr{T}_{3}\ast\mathscr{T}_{2}\ast\mathscr{T}_{1}(\mathcal{E}_{\mathrm{I}^{\otimes n_1}\otimes P_{2,x} \otimes P_{3}}) = \mathbb{E}_{Q_{2} \in \qty{\mathrm{I},\mathrm{Z}}^{\otimes n_2}}\mathbb{E}_{Q_3\in\mathcal{P}_{n_3}}[\mathcal{E}_{\mathrm{I}^{\otimes n_1} \otimes P_{2,x}Q_2 \otimes Q_3}]. \label{eq_twirling_case4}
\end{equation}

By combining Eqs.~\eqref{eq_superchannel_decomposition}, \eqref{eq_twirling_case1}, \eqref{eq_twirling_case2}, \eqref{eq_twirling_case3}, and \eqref{eq_twirling_case4}, we arrive at the conclusion of Theorem~\ref{thm_2_supply}.
\end{proof}

\red{\section{Symmetric Clifford twirling for general noise channels}
So far, we have assumed that the noise channel $\mathcal{N}$ following the non-Clifford layer $\mathcal{U}$ is Pauli noise.
While this assumption holds for the T gate or the Toffoli gate—both of which belong to the third level of the Clifford hierarchy~\cite{gottesman1999quantum} and can thus be Pauli-twirled into Pauli noise~\cite{cai2023quantum}—general non-Clifford unitaries cannot, in general, be Pauli-twirled.
In such cases, the noise channel $\mathcal{N}$ is not necessarily Pauli noise.
In this section, we analyze the performance of symmetric Clifford twirling under a general noise channel.}

\red{For simplicity, we consider the case where the non-Clifford layer $\mathcal{U}(\cdot) = U\cdot U^\dag$ is a Pauli-Z rotation gate applied to the first qubit, i.e., $U = e^{i\theta \mathrm{Z}\otimes \mathrm{I}^{\otimes n-1}}$, which is the main focus of this work.
A natural implementation of the Pauli-Z rotation gate is to synthesize it using Hadamard, S, and T gates.
Since we can Pauli-twirl all of these basic gates such that their associated noise becomes Pauli noise, we can ensure that the resulting noise channel for the Pauli-Z rotation gate is unital.
In this case, it is known that by twirling the Pauli-Z rotation gate using $\mathrm{S}^i$ gate with $i=0,1,2,3$, the noise channel $\mathcal{N}$ can be expressed as a composition of Pauli noise and coherent noise due to over-rotation~\cite{kwon2025criteria}:
$\mathcal{N}=\mathcal{N}_{\mathrm{Pauli}}\circ\mathcal{N}_{\mathrm{coh}}$, where
\begin{eqnarray}
    \mathcal{N}_{\mathrm{Pauli}}
    &=& (1-p_{x,y}-p_z)\mathcal{I} + \frac{p_{x,y}}{2}(\mathcal{E}_{\mathrm{X}\otimes \mathrm{I}^{\otimes n-1}} + \mathcal{E}_{\mathrm{Y}\otimes \mathrm{I}^{\otimes n-1}}) +p_z\mathcal{E}_{\mathrm{Z}\otimes \mathrm{I}^{\otimes n-1}},\\
    \mathcal{N}_{\mathrm{coh}}(\cdot) &=& e^{i\phi \mathrm{Z}\otimes \mathrm{I}^{\otimes n-1}} \cdot e^{-i\phi \mathrm{Z}\otimes \mathrm{I}^{\otimes n-1}}.
\end{eqnarray}}

\red{Since the symmetric Clifford operator $D\in\mathcal{G}_{n,\mathcal{Q}_{U}}$ commutes with the coherent noise $\mathcal{N}_{\mathrm{coh}}$, the coherent part remains unchanged under symmetric Clifford twirling.
Nevertheless, as we have discussed so far, the Pauli-X and Y components of the Pauli noise $\mathcal{N}_{\mathrm{Pauli}}$ can be scrambled through twirling:
\begin{equation}
    \mathscr{T}_{\mathcal{Q}_{U}}(\mathcal{N}) = \mathscr{T}_{\mathcal{Q}_{U}}(\mathcal{N}_{\mathrm{Pauli}}\circ\mathcal{N}_{\mathrm{coh}}) = \mathscr{T}_{\mathcal{Q}_{U}}(\mathcal{N}_{\mathrm{Pauli}})\circ\mathcal{N}_{\mathrm{coh}}
\end{equation}
with
\begin{equation}
    \label{eq_twirl_generalnoise}
    \mathscr{T}_{\mathcal{Q}_{U}}(\mathcal{N}_{\mathrm{Pauli}}) = 
    (1-p_{x,y}-p_z)\mathcal{I} + p_{x,y}\underset{\substack{Q_1\in\qty{\mathrm{X}, \mathrm{Y}} Q_2\in\mathcal{P}_{n-1}}}{\mathbb{E}}\qty[\mathcal{E}_{Q_1\otimes Q_2}] + p_z\mathcal{E}_{\mathrm{Z}\otimes \mathrm{I}^{\otimes n-1}}.
\end{equation}
Since the second term in $\mathscr{T}_{\mathcal{Q}_{U}}(\mathcal{N}_{\mathrm{Pauli}})$ represents global noise, this result highlights the effectiveness of symmetric Clifford twirling for general noise channels.
We note that the twirling of the Pauli-Z rotation gate using the S gate can be incorporated into the symmetric Clifford twirling procedure, as such operations are already included in the symmetric Clifford twirling.}

\section{Unitarity and the average noise strength}
In order to introduce the notion of white-noise approximation~\cite{dalzell2021random}, let us define unitarity and the average noise strength for the noise layer $\mathcal{N}$ in this section.

Let $\mathcal{N} = \sum_i E_i\cdot E_i^\dagger$ be a unital noise.
Then, we define the {\it unitarity} $u$~\cite{wallman2015estimating, carignan2019bounding, dalzell2021random} and the {\it average noise strength} $s$ of the noise channel as
\begin{eqnarray}
    \label{eq_unitarity}
    u 
    &=& \frac{2^n}{2^n-1}\qty(\mathbb{E}_{V\sim\mu_H}\qty[\tr[\mathcal{N}(V\ketbra{\psi}V^\dagger)^2]] - \frac{1}{2^n}),\\
    \label{eq_avnoisestr}
    s 
    &=& \frac{2^n}{2^n-1}\qty(\mathbb{E}_{V\sim\mu_H} \qty[\tr[V\ketbra{\psi}V^\dagger\mathcal{N}(V\ketbra{\psi}V^\dagger)]] - \frac{1}{2^n}),
\end{eqnarray}
where $\mathbb{E}_{V\sim\mu_H}$ denotes the average over the Haar measure $\mu_\mathrm{H}$ on the unitary group.
The unitarity $u$ is the expected purity of the output state under a random choice of input state, which has a minimum value of 0 and a maximum value of 1.
The average noise strength $s$ represents the noise strength of the twirled noisy channel as $\mathbb{E}_{V\sim\mu_H}[V^\dagger\mathcal{N}(V\cdot V^\dagger)V] = s\cdot + (1-s) I/2^n$~\cite{emerson2005scalable}.
We can represent the unitarity $u$ and the average noise strength $s$ using the Kraus operators of $\mathcal{N} = \sum_i E_i\cdot E_i^\dagger$ by using the 2-moment operator~\cite{mele2023introduction}
\begin{eqnarray}
    \mathcal{M}(\mathbb{O}) 
    &=& \mathbb{E}_{V \sim \mu_H}[V^{\otimes 2} \mathbb{O} V^{\dagger\otimes 2}] \\
    &=& \frac{2^{2n}\mathrm{tr}[\mathbb{O}] - 2^n\mathrm{tr}[\mathbb{O} \mathbb{F}] }{2^{2n}-1} \frac{\mathbb{I}}{2^{2n}} + \frac{2^{n}\mathrm{tr}[\mathbb{O} \mathbb{F}] - \mathrm{tr}[\mathbb{O}] }{2^{2n}-1} \frac{\mathbb{F}}{2^n},
    \label{eq_moment_op}
\end{eqnarray}
where $\mathbb{O}$ is an arbitrary operator, $\mathbb{I}$ is the identity operator and $\mathbb{F}$ is the swap operator on the two copies of $2^{n}$-dimensional Hilbert space.
By using the 2-moment operator $\mathcal{M}(\mathbb{O})$, we obtain
\begin{eqnarray}
    u
    &=& \frac{2^n}{2^n-1}\qty(\sum_{ij}\tr[(E_i \otimes E_j) \mathcal{M}((\ketbra{\psi})^{\otimes 2}) (E_i^\dagger \otimes E_j^\dagger)\mathbb{F}] - \frac{1}{2^n})\\
    &=& \frac{2^n}{2^n-1}\qty(\sum_{ij}\tr[\qty(E_i \otimes E_j) \qty(\frac{2^{2n} - 2^n}{2^{2n}-1} \frac{\mathbb{I}}{2^{2n}} + \frac{1}{2^n+1}\frac{\mathbb{F}}{2^n}) (E_i^\dagger \otimes E_j^\dagger)\mathbb{F}] - \frac{1}{2^n})\\
    &=& \frac{2^n}{2^n-1}\qty(\frac{2^{n} - 1}{2^n(2^{2n}-1)}\sum_{ij}\mathrm{tr}[E_iE_i^\dagger E_jE_j^\dagger] + \frac{1}{2^n(2^n+1)}\sum_{ij}\mathrm{tr}|[E_iE_j^\dagger]|^2 - \frac{1}{2^n})\\
    &=& \frac{\sum_{ij} |\mathrm{tr}[E_iE_j^\dagger]|^2-1}{4^n-1}
\end{eqnarray}
for the unitarity $u$ and 
\begin{eqnarray}
    s
    &=& \frac{2^n}{2^n-1}\qty(\sum_{i}\tr[(E_i \otimes I) \mathcal{M}((\ketbra{\psi})^{\otimes 2}) (E_i^\dagger \otimes I)\mathbb{F}] - \frac{1}{2^n})\\
    &=& \frac{2^n}{2^n-1}\qty(\sum_{i}\tr[\qty(E_i \otimes I) \qty(\frac{2^{2n} - 2^n}{2^{2n}-1} \frac{\mathbb{I}}{2^{2n}} + \frac{1}{2^n+1}\frac{\mathbb{F}}{2^n}) (E_i^\dagger \otimes I)\mathbb{F}] - \frac{1}{2^n})\\
    &=& \frac{2^n}{2^n-1}\qty(\frac{2^{n} - 1}{2^n(2^{2n}-1)}\sum_{i}\mathrm{tr}[E_iE_i^\dagger] + \frac{1}{2^n(2^n+1)}\sum_{i}|\mathrm{tr}[E_i]|^2 - \frac{1}{2^n})\\
    &=& \frac{\sum_i |\mathrm{tr}[E_i]|^2 - 1}{4^n-1}
\end{eqnarray}
for the average noise strength $s$, where we used $\sum_i E_iE_i^\dagger = I$ for unital noise.
Thus, for Pauli noise $\mathcal{N}(\cdot) = (1-p_{\mathrm{err}})\cdot + \sum_i p_i P_i\cdot P_i$, these parameters can be represented as
\begin{eqnarray}
    \label{eq_unitarity_Pauli_}
    u &=& 1 - \frac{4^n}{4^n-1}2p_{\mathrm{err}} + \frac{4^n}{4^n-1}\qty(p_{\mathrm{err}}^2 + \red{\sum_{i=1}^{4^n-1}} p_i^2),\\
    \label{eq_avnoisestr_Pauli_}
    s &=& 1 - \frac{4^n}{4^n-1}p_{\mathrm{err}}.
\end{eqnarray}

\section{White-noise approximation in the early-FTQC regime}
In the realm of NISQ computing, it has been demonstrated that the effective noise of random quantum circuits consisting of noisy 2-qubit gates can be approximated as the global white  noise~\cite{dalzell2021random}.
This approximation, known as the {\it white-noise approximation}, arises from the scrambling of local errors to other qubits due to the randomness of gates.
In the early FTQC regime, the primary source of errors is not multi-qubit Clifford gates, but rather (potentially single-qubit) non-Clifford gates.
Conversely, as long as the gate is Clifford, we can execute multi-qubit operations with negligible errors.
This implies that the Clifford layer $\mathcal{C}_{l}(\cdot) = C_l\cdot C_l^\dagger$ can, in theory, represent any $n$-qubit Clifford unitary.
Thus, if we want to capture the typical behavior of the logical circuits, we may consider the case where $C_l$ is randomly and uniformly drawn from the $n$-qubit Clifford group $\mathcal{G}_n$.
In such a scenario, local noise is anticipated to be more effectively scrambled to other qubits compared to NISQ circuits, where randomness arises from noisy 2-qubit gates. 
Therefore, we can assert that the white-noise approximation is better suited for the early FTQC regime than the NISQ regime.

As in Sec.~S2, we focus on the scenario where we aim to estimate the expectation value of an observable $O$ for the $n$-qubit 
noiseless quantum state
\begin{equation}
    \rho_{\mathrm{ideal}} := \mathcal{C}_{L+1}\circ\mathcal{U}\circ\mathcal{C}_{L}\circ \cdots \circ\mathcal{U}\circ\mathcal{C}_{1}(\rho_0),
\end{equation}
but we only have access to the noisy logical quantum state
\begin{equation}
    \rho_{\mathrm{noisy}} := \mathcal{C}_{L+1}\circ\mathcal{N}\circ\mathcal{U}\circ\mathcal{C}_{L}\circ \cdots \circ\mathcal{N}\circ\mathcal{U}\circ\mathcal{C}_{1}(\rho_0),
\end{equation}
where $\mathcal{C}_l(\cdot) := C_l \cdot C_l^\dagger$ is the $l$-th Clifford layer, $\mathcal{U}(\cdot) := U \cdot U^\dagger$ is the non-Clifford layer, $\mathcal{N}$ is the noise layer, $\rho_0$ is the initial state, and $L$ represents the number of the noisy non-Clifford layers.
Here, for simplicity, we assume that $O$ is a traceless observable satisfying $-I\leq O \leq I$ and $\mathcal{N}$ is a unital noise.

When we assume white-noise approximation, we can estimate the ideal expectation value $\mathrm{tr}[\rho_{\mathrm{ideal}}O]$ by rescaling the noisy expectation value $\mathrm{tr}[\rho_{\mathrm{noisy}}O]$ by some factor $R$.
As the performance of white-noise approximation in the early FTQC regime for mitigating errors, we derive the following theorem by utilizing the discussion of Sec.~5.2 of Ref.~\cite{dalzell2021random}.
\begin{thm}
    \label{thm_S1}
    Let each Clifford gate $C_l$ be drawn randomly and uniformly from the $n$-qubit Clifford group $\mathcal{G}_n$.
    Then, the noisy expectation value rescaled by the constant factor $R = (s/u)^L$ satisfies
    \begin{eqnarray}
        \label{eq_bias_general}
        \mathbb{E}_{C}[|R\mathrm{tr}[\rho_{\mathrm{noisy}}O] - \mathrm{tr}[\rho_{\mathrm{ideal}}O]|]
        \leq \sqrt{\frac{2^n-1}{2^n+1} \qty(1 - \qty(\frac{s^2}{u})^L )},
    \end{eqnarray}
    where $\mathbb{E}_{C}$ represents the uniform average over all Clifford unitary $C_l$ over $\mathcal{G}_n$. Here, we assumed that $\Vert O\Vert \leq 1$.
\end{thm}

The sampling overhead $N$ required to obtain the rescaled noisy expectation value $R\mathrm{tr}[\rho_{\mathrm{noisy}}O]$ satisfies
\begin{equation}
    \label{eq_cost_rescale}
    N = \qty(\frac{s}{u})^{2L} \sim \qty(1+ \frac{4^n}{4^n-1}p_{\mathrm{err}})^{2L}
\end{equation}
up to the first degree of $p_{\mathrm{err}}$ for Pauli noise $\mathcal{N} = (1-p_{\mathrm{err}})\mathcal{I} + \sum_i p_i\mathcal{E}_{P_i}$, where Pauli error $\mathcal{E}_{P_i}(\cdot) := P_i\cdot P_i$ occurs with probability $p_i$ and $p_{\mathrm{err}} = \sum_i p_i$ being the error probability.
This saturates the lower bound on the sampling cost represented as Eq. (\ref{eq_QEMbound}) when $p_{\mathrm{err}}$ is small and $n$ is large, thus validating this rescaling method as a cost-optimal QEM method.
This means that, besides converting local noise into global white noise through (symmetric) Clifford twirling, we can achieve cost-optimal QEM by assuming the white-noise approximation.

While the method described above allows us to mitigate errors with minimal sampling overhead, it still grows exponentially with $L$.
Therefore, a practical scenario arises when the total error rate $p_{\mathrm{tot}} := p_{\mathrm{err}}L$ is maintained at a constant value, resulting in a constant sampling overhead approximated as $e^{2p_{\mathrm{tot}}}$ for large $L$.
In this case, we obtain the following corollary for Pauli noise:
\begin{cor}
    \label{cor_1}
    Let each Clifford gate $C_l$ be randomly and uniformly drawn from the $n$-qubit Clifford group $\mathcal{G}_n$.
    Then, the noisy expectation value rescaled by the constant factor $R = (s/u)^L$ satisfies
    \begin{equation}
        \label{eq_bias_const}
        \mathbb{E}_{C}[|R\mathrm{tr}[\rho_{\mathrm{noisy}}O] - \mathrm{tr}[\rho_{\mathrm{ideal}}O]|] \lesssim   \frac{vp_{\mathrm{tot}}}{\sqrt{L}}
    \end{equation}   
    for large $L$ when $p_{\mathrm{tot}} = p_{\mathrm{err}}L$ is fixed to a constant value, where
    \begin{equation}
        v = \sqrt{\red{\sum_{i=1}^{4^n-1}}\qty(\frac{p_i}{p_{\mathrm{err}}} - \frac{1}{4^n-1})^2}
    \end{equation}
    represents the distance between the Pauli noise $\mathcal{N}$ and the global white noise. Here, we assumed that $\Vert O\Vert \leq 1$.
\end{cor}

In the early FTQC regime, the logical error probability $p_{\mathrm{err}}$ is expected to decrease as hardware improves, allowing for an increase in the number of available noisy non-Clifford operations $L$ for a fixed sampling overhead. Corollary \ref{cor_1} indicates that, in this scenario, the bias obtained through rescaling the noisy expectation value diminishes to 0, thereby confirming the effectiveness of the white-noise approximation as the hardware improves in the early FTQC regime.
Even if $L$ may not be sufficiently large to achieve the desired accuracy, we can reduce the average bias by reducing the parameter $v$ using symmetric Clifford twirling.

While Corollary \ref{cor_1} assumes that each Clifford gate $\mathcal{C}_l$ is randomly sampled from the global $n$-qubit Clifford unitary, we anticipate a similar scaling behavior of $|R\mathrm{tr}[\rho_{\mathrm{noisy}}O] - \mathrm{tr}[\rho_{\mathrm{ideal}}O]| \sim O(v/\sqrt{L})$ holds even for simpler structured circuits, such as the Trotterized Hamiltonian simulation circuits demonstrated in the main text.
Let us denote the Trotter step size as $T$.
Then, we can represent $L$ as $L = nMT$, where $M$ is the number of terms in the Hamiltonian.
Therefore, we expect the scaling of $|R\mathrm{tr}[\rho_{\mathrm{noisy}}O] - \mathrm{tr}[\rho_{\mathrm{ideal}}O]| \sim O(v/\sqrt{nT})$ to hold for Trotterized Hamiltonian simulation circuits.

\begin{figure*}[t]
    \begin{center}
        \resizebox{0.99\hsize}{!}{\includegraphics{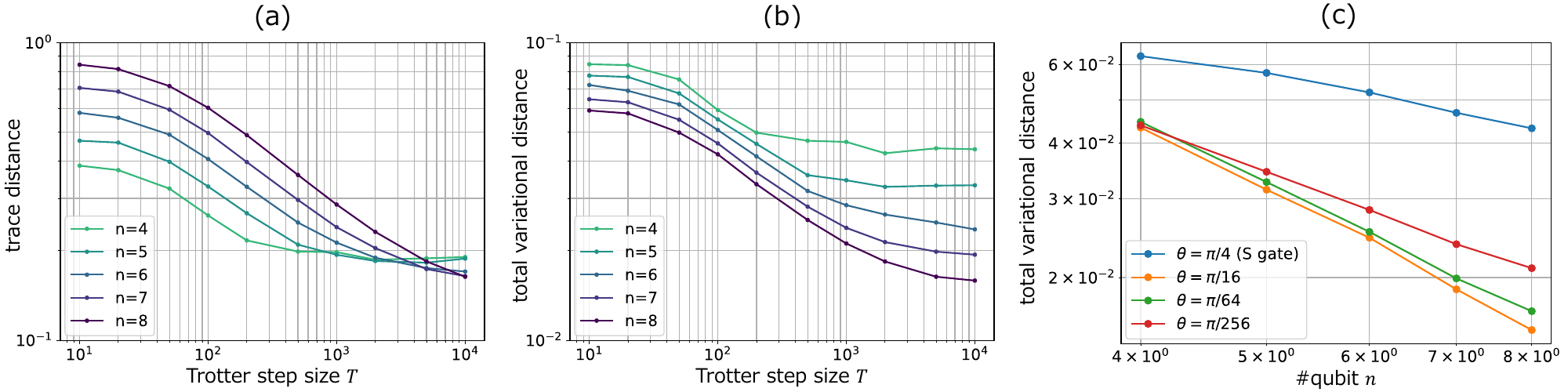}}
        \caption{Performance of white-noise approximation for Trotterized Hamiltonian simulation circuit of 1D Heisenberg model. Panel (a) represents the trace distance between the ideal noiseless state $\rho_{\mathrm{ideal}}$ and the rescaled virtual noisy state $R\rho_{\mathrm{noisy}}+(1-R)I/2^n$ averaged over random input state. Panel (b) and (c) represent the total variational distance between the probabilistic distribution obtained through measurements on the ideal noiseless state $\rho_{\mathrm{ideal}}$ and the rescaled virtual noisy state $R\rho_{\mathrm{noisy}}+(1-R)I/2^n$, averaged over random input state and measurement bases.
        Panel (a) and (b) show the Trotter step size $L$ dependence for different qubit count $n$ with a fixed rotation angle $\theta = \pi/256$, while Panel (c) shows the qubit count $n$ dependence for different rotation angle $\theta$ with a fixed Trotter step size $L = 1000$. We fix the total error rate as $p_{\mathrm{tot}} := p_{\mathrm{err}}nMT = 1$.
        }
        \label{fig_numerics_supply}
    \end{center}
\end{figure*}

Let us first compare this analysis with the large-scale Clifford simulation shown in Fig.~4 in the main text.
In this figure, we observe that the average bias scales as $1/\sqrt{n}$ for a fixed $T$, which is consistent with our theoretical analysis.
Furthermore, we find that the scaling improves to $1/n$ when applying 2-sparse symmetric Clifford twirling.
This improvement can be attributed to the reduction of $v$ from $1$ to $O(1/\sqrt{n})$.

Next, we compare the above analysis with the non-Clifford state simulation for the Trotterized Hamiltonian simulation circuit of the 1D Heisenberg model under the open boundary condition, depicted in Fig.~\ref{fig_numerics_supply}.
In this setup, we assume that the noisy non-Clifford layer is given as $\mathcal{N} \circ \mathcal{U}$ instead of the noiseless $\mathcal{U}(\cdot) = U\cdot U^\dag$, where $U = e^{i\theta \mathrm{Z}\otimes \mathrm{I}^{\otimes n-1}}$ represents a Pauli-Z rotation gate.
We sandwich this non-Clifford layer with some Clifford operations to give the target Pauli rotation.
As for the noise, we assume that $\mathcal{N}$ represents depolarizing noise.
We set the total error rate to $p_{\mathrm{tot}} = 1$, which results in the constant sampling overhead of $e^2\sim7$.

Fig.~\ref{fig_numerics_supply}(a) represents the trace distance between the ideal noiseless state $\rho_{\mathrm{ideal}}$ and the rescaled virtual noisy state $R\rho_{\mathrm{noisy}}+(1-R)I/2^n$, averaged over random input states.
We note that the trace distance provides an upper bound on the bias when mitigating errors for arbitrary observables.
The trace distance decreases as the Trotter step size $T$ increases, but eventually saturates to a fixed value.
Additionally, the trace distance generally increases with the qubit count $n$ for small $T$.
These results contradict the above analysis, suggesting that the white-noise approximation does not hold in the worst case (i.e., for arbitrary observables) for Trotterized Hamiltonian simulation circuits.

However, we can justify the approximation in the average case, as shown in Fig.~\ref{fig_numerics_supply}(b) and (c).
These figures display the total variational distance between the probability distributions obtained from measurements on the ideal noiseless state $\rho_{\mathrm{ideal}}$ and the rescaled virtual noisy state $R\rho_{\mathrm{noisy}}+(1-R)I/2^n$, averaged over random input states and measurement bases.
Although the total variational distance also saturates as we increase the Trotter step size $T$, it exhibits a decrease of $O(1/\sqrt{n})$ as the qubit count $n$ grows.
Since the total variational distance provides an upper bound on the bias for observables diagonalizable by the measurement bases, these results suggest a scaling of $|R\mathrm{tr}[\rho_{\mathrm{noisy}}O] - \mathrm{tr}[\rho_{\mathrm{ideal}}O]| \sim O(1/\sqrt{n})$ for typical observables.
This average scaling is especially useful when we estimate expectation value through shadow tomography~\cite{huang2020predicting}, where we perform measurements on random basis.

Therefore, we expect the white-noise approximation to hold on average in the early FTQC regime, where quantum circuits are sufficiently deep.
It is worth noting that this property is specific to the early FTQC regime, as the white-noise approximation does not apply to shallow structured quantum circuits in the NISQ regime~\cite{foldager2023can}.
We also note that the validity of the Clifford simulation shown in Fig.~4 of the main text can be supported by Fig.~\ref{fig_numerics_supply}(c), as the qualitative behavior of the total variational distance is unaffected by the rotation angle $\theta$ of the Pauli-Z rotation gate $R_z(\theta)$.

\section{Proof of Theorem~\ref{thm_S1}}
In this section, we present the proof of Theorem~\ref{thm_S1}.
This proof mainly follows the discussion of Sec.~5.2 of Ref.~\cite{dalzell2021random}.

\begin{proof}
Since the $n$-qubit Clifford group $\mathcal{G}_n$ is a unitary 2-design and we only consider up to a second moment in the following proof, we can assume that each Clifford gate $C_l$ is sampled randomly from $n$-qubit unitary group according to the Haar measure.
Since the Haar measure on the unitary group is invariant under the right and left multiplication of unitary gates, we can insert any unitary gates before and after $C_l$, so we can neglect the non-Clifford layer $\mathcal{U}$ and assume that the observable $O$ can be diagonalized using the computational basis $\qty{\ket{x}}_x$ as $O = \sum_x o_x \ketbra{x}$, where $-1\leq o_x\leq 1$ and $\sum_x o_x = 0$.
Therefore, by defining the probability distributions $p_{\mathrm{ideal}}$ and $p_{\mathrm{noisy}}$ as
\begin{eqnarray}
    p_{\mathrm{ideal}}(x) &=& \mathrm{tr}[\ketbra{x}\mathcal{C}_{L+1}\circ\mathcal{C}_L\circ\cdots\circ \mathcal{C}_1(\ketbra{0^n})],\\ 
    p_{\mathrm{noisy}}(x) &=& \mathrm{tr}[\ketbra{x}\mathcal{C}_{L+1}\circ\mathcal{N}\circ\mathcal{C}_L\circ\cdots\circ \mathcal{N}\circ\mathcal{C}_1(\ketbra{0^n})],
\end{eqnarray}
we can upper-bound the average bias as
\begin{eqnarray}
    \mathbb{E}_{C}[|R\mathrm{tr}[\rho_{\mathrm{noisy}}O] - \mathrm{tr}[\rho_{\mathrm{ideal}}O]|]
    &=& \mathbb{E}_{C}\qty[\abs{\sum_x o_x\qty(Rp_{\mathrm{noisy}}(x) + (1-R)2^{-n}  - p_{\mathrm{ideal}}(x))}] \\
    &\leq& \mathbb{E}_{C}\qty[\sum_x\abs{(Rp_{\mathrm{noisy}}(x) + (1-R)2^{-n} - p_{\mathrm{ideal}}(x))}] \\
    &\leq& \mathbb{E}_{C}\qty[2^{n/2} \sqrt{\sum_x(Rp_{\mathrm{noisy}}(x) + (1-R)2^{-n} - p_{\mathrm{ideal}}(x))^2}] \\
    &\leq& \sqrt{2^n\mathbb{E}_{C}\qty[\sum_x(Rp_{\mathrm{noisy}}(x) + (1-R)2^{-n} - p_{\mathrm{ideal}}(x))^2]},
\end{eqnarray}
where we used $\sum_x o_x = 0$ in the first equality, the triangle inequality and $|o_x|\leq 1$ in the second inequality, $\norm{\mathbf{v}}_1 \leq 2^{n/2} \norm{\mathbf{v}}_2$ for $2^n$-dimensional vector $\mathbf{v}$ in the third inequality, and Jensen's inequality in the last inequality.
Then, by defining
\begin{eqnarray}
    Z_0 &=& 2^n \mathbb{E}_{C} \qty[\sum_xp_{\mathrm{ideal}}(x)^2] \\
    &=& 2^{2n} \mathbb{E}_{C} \qty[p_{\mathrm{ideal}}(0^n)^2], \\
    Z_1 &=& 2^n \mathbb{E}_{C} \qty[\sum_xp_{\mathrm{ideal}}(x)p_{\mathrm{noisy}}(0^n)] \\
    &=& 2^{2n} \mathbb{E}_{C} \qty[p_{\mathrm{ideal}}(0^n)p_{\mathrm{noisy}}(0^n)], \\
    Z_2 &=& 2^n \mathbb{E}_{C} \qty[\sum_xp_{\mathrm{noisy}}(x)^2] \\
    &=& 2^{2n} \mathbb{E}_{C} \qty[p_{\mathrm{noisy}}(0^n)^2],
\end{eqnarray}
we obtain
\begin{eqnarray}
    \mathbb{E}_{C}[|R\mathrm{tr}[\rho_{\mathrm{noisy}}O] - \mathrm{tr}[\rho_{\mathrm{ideal}}O]|]
    &\leq& \sqrt{2^n\mathbb{E}_{C}\qty[\sum_x(R(p_{\mathrm{noisy}}(x) -2^{-n}) - (p_{\mathrm{ideal}}(x) -2^{-n}))^2]} \\
    &=& \sqrt{R^2 (Z_2-1) - 2R (Z_1-1) + (Z_0 - 1)},
\end{eqnarray}
where the RHS is minimized when we take $R = \frac{Z_1-1}{Z_2-1}$.
In this case, we obtain
\begin{equation}
    \mathbb{E}_{C}[|R\mathrm{tr}[\rho_{\mathrm{noisy}}O] - \mathrm{tr}[\rho_{\mathrm{ideal}}O]|] \leq \sqrt{ (Z_0-1)\qty(1 - \frac{(Z_1-1)^2}{(Z_2-1)(Z_0-1)})}.
\end{equation}

Now, let us analyze $Z_0$, $Z_1$, and $Z_2$.
By using the 2-moment operator defined in Eq. (\ref{eq_moment_op}), we obtain
\begin{eqnarray}
    Z_0 
    &=& 2^{2n} \mathbb{E}_{C} \qty[\mathrm{tr}[\ketbra{0^n}\mathcal{C}_{L+1}\circ\mathcal{C}_L\circ\cdots\circ \mathcal{C}_1(\ketbra{0^n})]^2]\\
    &=& 2^{2n} \mathbb{E}_{C} \qty[\mathrm{tr}[(\ketbra{0^n}\mathcal{C}_{L+1}\circ\mathcal{C}_L\circ\cdots\circ \mathcal{C}_1(\ketbra{0^n}))^{\otimes 2} ]]\\
    &=& 2^{2n} \mathbb{E}_{C} \qty[\mathrm{tr}[\ketbra*{0^{2n}}\mathcal{C}_{L+1}^{\otimes 2}\circ\mathcal{C}_L^{\otimes 2}\circ\cdots\circ \mathcal{C}_1^{\otimes 2}(\ketbra*{0^{2n}})]]\\
    &=& 2^{2n}\mathrm{tr}[\ketbra*{0^{2n}}\underbrace{\mathcal{M} \circ \cdots \circ \mathcal{M}}_{L+1}(\ketbra*{0^{2n}}) ]].
\end{eqnarray}
In the same way, we obtain
\begin{eqnarray}
    \label{eq_z_moment_op}
    Z_i = 2^{2n}\mathrm{tr}[\ketbra*{0^{2n}}\underbrace{\mathcal{M} \circ \mathcal{N}_i \circ \cdots \circ \mathcal{M} \circ \mathcal{N}_i}_{L} \circ \mathcal{M}(\ketbra*{0^{2n}}) ]]
\end{eqnarray}
for $i = 0, 1, 2$, where we have defined 
\begin{eqnarray}
    \mathcal{N}_0 &=& \mathcal{I} \otimes \mathcal{I}, \\
    \mathcal{N}_1 &=& \mathcal{I} \otimes \mathcal{N}, \\
    \mathcal{N}_2 &=& \mathcal{N} \otimes \mathcal{N}
\end{eqnarray}
using the identity channel $\mathcal{I}$.

We next analyze Eq. (\ref{eq_z_moment_op}).
From Eq. (\ref{eq_moment_op}), we have
\begin{eqnarray}
    \mathcal{M}(\ketbra*{0^{2n}})  = \qty(1 - \frac{1}{2^{n}+1}) \frac{\mathbb{I}}{2^{2n}} + \frac{1}{2^{n}+1} \frac{\mathbb{F}}{2^n}.
\end{eqnarray}
Next, we calculate the action of $\mathcal{M}\circ\mathcal{N}_i$ on $\mathbb{I}/2^{2n}$ and $\mathbb{F}/2^{n}$.
Since $\mathcal{I}$ and $\mathcal{N}$ are both unital, we obtain
\begin{eqnarray}
    \mathcal{M}\circ\mathcal{N}_i\qty(\frac{\mathbb{I}}{2^{2n}}) = \frac{\mathbb{I}}{2^{2n}}
\end{eqnarray}
for $i = 0, 1, 2$.
For the input state $\mathbb{F}/2^{n}$, we obtain
\begin{eqnarray}
    \mathcal{M}\circ \mathcal{N}_0 \qty(\frac{\mathbb{F}}{2^{n}}) 
    &=& \mathcal{M} \qty(\frac{\mathbb{F}}{2^{n}}) \\
    &=& \frac{\mathbb{F}}{2^{n}}, \\
    \mathcal{M}\circ \mathcal{N}_1 \qty(\frac{\mathbb{F}}{2^{n}})
    &=& 2^{-n} \mathcal{M}(\mathcal{I} \otimes \mathcal{N}(\mathbb{F})) \\
    &=& \frac{2^{n}\mathrm{tr}[\mathcal{I} \otimes \mathcal{N}(\mathbb{F})] - \mathrm{tr}[\mathcal{I} \otimes \mathcal{N}(\mathbb{F}) \mathbb{F}] }{2^{2n}-1} \frac{\mathbb{I}}{2^{2n}} + \frac{\mathrm{tr}[\mathcal{I} \otimes \mathcal{N}(\mathbb{F}) \mathbb{F}] - 2^{-n}\mathrm{tr}[\mathcal{I} \otimes \mathcal{N}(\mathbb{F})] }{2^{2n}-1} \frac{\mathbb{F}}{2^n}\\
    &=& (1-s)\frac{\mathbb{I}}{2^{2n}} + s \frac{\mathbb{F}}{2^{2n}},\\
    \mathcal{M}\circ \mathcal{N}_2 \qty(\frac{\mathbb{F}}{2^{n}})
    &=& 2^{-n} \mathcal{M}(\mathcal{N} \otimes \mathcal{N}(\mathbb{F})) \\
    &=& \frac{2^{n}\mathrm{tr}[\mathcal{N} \otimes \mathcal{N}(\mathbb{F})] - \mathrm{tr}[\mathcal{N} \otimes \mathcal{N}(\mathbb{F}) \mathbb{F}] }{2^{2n}-1} \frac{\mathbb{I}}{2^{2n}} + \frac{\mathrm{tr}[\mathcal{N} \otimes \mathcal{N}(\mathbb{F}) \mathbb{F}] - 2^{-n}\mathrm{tr}[\mathcal{N} \otimes \mathcal{N}(\mathbb{F})] }{2^{2n}-1} \frac{\mathbb{F}}{2^n}\\
    &=& (1-u)\frac{\mathbb{I}}{2^{2n}} + u \frac{\mathbb{F}}{2^{2n}},
\end{eqnarray}
where we have used
\begin{eqnarray}
    \mathrm{tr}[\mathcal{I} \otimes \mathcal{N}(\mathbb{F})] 
    &=& \sum_i \mathrm{tr}[(I\otimes E_i) \mathbb{F} (I\otimes E_i^\dagger)] \\
    &=& \sum_i \mathrm{tr}[E_iE_i^\dagger] \\
    &=& 2^n,\\
    \mathrm{tr}[\mathcal{I} \otimes \mathcal{N}(\mathbb{F})\mathbb{F}] 
    &=& \sum_i \mathrm{tr}[(I\otimes E_i) \mathbb{F} (I\otimes E_i^\dagger) \mathbb{F}] \\
    &=& \sum_i \mathrm{tr}[E_i]\mathrm{tr}[E_i^\dagger] \\
    &=& (2^{2n}-1)s + 1, \\
    \mathrm{tr}[\mathcal{N} \otimes \mathcal{N}(\mathbb{F})] 
    &=& \sum_{ij} \mathrm{tr}[(E_j\otimes E_i) \mathbb{F} (E_j^\dagger \otimes E_i^\dagger)] \\
    &=& \sum_{ij} \mathrm{tr}[E_i E_j^\dagger E_j E_i^\dagger] \\
    &=& 2^n,\\
    \mathrm{tr}[\mathcal{N} \otimes \mathcal{N}(\mathbb{F})\mathbb{F}] 
    &=& \sum_i \mathrm{tr}[(E_j\otimes E_i) \mathbb{F} (E_j\otimes E_i^\dagger) \mathbb{F}] \\
    &=& \sum_i \mathrm{tr}[E_iE_j^\dagger]\mathrm{tr}[E_i^\dagger E_j] \\
    &=& (2^{2n}-1)u + 1.
\end{eqnarray}
Therefore, we obtain
\begin{eqnarray}
    Z_0 
    &=& 2^{2n}\qty(\qty(1 - \frac{1}{2^{n}+1}) \frac{\mathrm{tr}[\ketbra*{0^{2n}}\mathbb{I}]}{2^{2n}} + \frac{1}{2^{n}+1} \frac{\mathrm{tr}[\ketbra*{0^{2n}}\mathbb{F}]}{2^n}) \\
    &=& 1+\frac{2^n-1}{2^n+1}, \\
    Z_1
    &=& 2^{2n}\qty(\qty(1 - \frac{1}{2^{n}+1}s^L) \frac{\mathrm{tr}[\ketbra*{0^{2n}}\mathbb{I}]}{2^{2n}} + \frac{1}{2^{n}+1}s^L \frac{\mathrm{tr}[\ketbra*{0^{2n}}\mathbb{F}]}{2^n}) \\
    &=& 1+\frac{2^n-1}{2^n+1}s^L, \\
    Z_2
    &=& 2^{2n}\qty(\qty(1 - \frac{1}{2^{n}+1}u^L) \frac{\mathrm{tr}[\ketbra*{0^{2n}}\mathbb{I}]}{2^{2n}} + \frac{1}{2^{n}+1}u^L \frac{\mathrm{tr}[\ketbra*{0^{2n}}\mathbb{F}]}{2^n}) \\
    &=& 1+\frac{2^n-1}{2^n+1}u^L.
\end{eqnarray}

Thus, we have
\begin{eqnarray}
    \label{eq_bias_proof}
    \mathbb{E}_{C}[|R\mathrm{tr}[\rho_{\mathrm{noisy}}O] - \mathrm{tr}[\rho_{\mathrm{ideal}}O]|]
    &\leq& \sqrt{\frac{2^n-1}{2^n+1} \qty(1 - \qty(\frac{s^2}{u})^L )}.
\end{eqnarray}
for $R = (s/u)^L$.
\end{proof}

For Pauli noise $\mathcal{N} = (1-p_{\mathrm{err}})\mathcal{I} + \sum_i p_i\mathcal{E}_{P_i}$, we can calculate the RHS of Eq.~(\ref{eq_bias_proof}) as
\begin{eqnarray}
    \sqrt{\frac{2^n-1}{2^n+1} \qty(1 - \qty(\frac{s^2}{u})^L )} 
    &=& \sqrt{\frac{2^n-1}{2^n+1} \qty(1 - \qty(1 - \frac{4^n}{4^n-1}\qty(\sum_i p_i^2 - \frac{1}{4^n-1}p_{\mathrm{err}}^2) + O(p_i^3))^L )}\\
    &=& \sqrt{\frac{2^n-1}{2^n+1} \qty(1 - \qty(1 - \frac{4^n}{4^n-1}v^2p_{\mathrm{err}}^2 + O(p_i^3))^L )},
\end{eqnarray}
where we use Eqs.~(\ref{eq_unitarity_Pauli_}) and (\ref{eq_avnoisestr_Pauli_}).
Especially when we assume that $p_{\mathrm{tot}} = p_{\mathrm{err}}L$ is constant and $L \gg 1$, we have
\begin{eqnarray}
    \sqrt{\frac{2^n-1}{2^n+1} \qty(1 - \qty(\frac{s^2}{u})^L )} 
    &\sim& \sqrt{\frac{2^n-1}{2^n+1} \qty(1 - \qty(1 - \frac{4^n}{4^n-1}\frac{v^2p_{\mathrm{tot}}^2}{L^2})^L )} \\
    &\sim& \sqrt{\frac{2^n-1}{2^n+1}\frac{4^n}{4^n-1}\frac{v^2p_{\mathrm{tot}}^2}{L}}\\
    &\leq& \frac{vp_{\mathrm{tot}}}{\sqrt{L}}
\end{eqnarray}
up to the leading order of $L$.
Therefore, we obtain Corollary~\ref{cor_1}.

\end{document}